%% file: TEXT.tex
\documentclass[a4paper,fleqn]{cas-dc}

\usepackage[authoryear,longnamesfirst]{natbib}

\def\tsc#1{\csdef{#1}{\textsc{\lowercase{#1}}\xspace}}
\tsc{WGM}
\tsc{QE}

\usepackage{graphicx}    
\usepackage{amssymb}
\usepackage{amsmath}
\usepackage{makecell}

\usepackage{textcomp}
\usepackage{stfloats}
\usepackage{url}
\usepackage{verbatim}
\usepackage{subcaption}
\usepackage{xcolor}
\newtheorem{theorem}{Theorem}
\newtheorem{definition}{Definition}
\newtheorem{lemma}{Lemma}

\newtheorem{remark}{Remark}

\usepackage{tikz}
\usetikzlibrary{positioning,calc}
\usepackage{tikzscale}
\usetikzlibrary{arrows, positioning, calc}
\usepackage{environ}

\newcommand{\RePart}[1]{\mathrm{Re}\{#1\}}
\newcommand{\ImPart}[1]{\mathrm{Im}\{#1\}}

\makeatletter
\newsavebox{\measure@tikzpicture}
\NewEnviron{scaletikzpicturetowidth}[1]{%
	\def\tikz@width{#1}%
	\def\tikzscale{1}\begin{lrbox}{\measure@tikzpicture}%
		\BODY
	\end{lrbox}%
	\pgfmathparse{#1/\wd\measure@tikzpicture}%
	\edef\tikzscale{\pgfmathresult}%
	\BODY
    }

\begin{document}
\let\WriteBookmarks\relax
\def\floatpagepagefraction{1}
\def\textpagefraction{.001}

\title [mode = title]{Revisiting the generalized first-order reset element with shaping filters}                      
\tnotemark[1]

\tnotetext[1]{This work is co-financed by ASMPT and Holland High Tech, Topsector High Tech Systems and Materials, with a PPS innovation grant public-private collaboration for research and development.}

\author[1]{Ali Hosseini}
\cormark[1]
\ead{s.a.hosseini@tudelft.nl}

\credit{Conceptualization of this study, Methodology, Software}

\affiliation[1]{organization={Department of Precision and Microsystems Engineering, Delft University of Technology},
            city={Delft},
            postcode={2628 CD}, 
            country={The Netherlands}}

\author[2]{Dragan Kosti\'c}
\ead{dragan.kostic@asmpt.com}

\author[1]{Hassan HosseinNia}
\ead{s.h.hosseinniakani@tudelft.nl}

\affiliation[2]{organization={ASMPT},
            city={Beuningen},
            postcode={6641 TL}, 
            country={The Netherlands}}

\cortext[cor1]{Corresponding author}


\begin{abstract}
Reset control provides a nonlinear approach for improving closed-loop performance beyond the limitations of linear time-invariant controllers. However, the reset action inevitably introduces higher-order harmonics, which may degrade tracking performance, distort the reset signal, and reduce the reliability of frequency-domain predictions obtained via describing-function analysis. This paper revisits the generalized first-order reset element with shaping filters and develops a systematic framework for suppressing undesired reset-induced nonlinearities. Analytical conditions are derived for shaping filter coefficients to increase the low-frequency attenuation slope of the magnitude of the higher-order sinusoidal input describing functions (HOSIDFs). By modifying the asymptotic attenuation behavior of these higher-order harmonics, the proposed design provides stronger harmonic suppression in frequency regions where reset action is undesired, while preserving the beneficial first-order harmonic phase advantage near the desired cross-over frequency. The reduction in nonlinear behavior is verified through HOSIDF analysis and a superposition-law test, demonstrating that higher-order shaping filters make the reset element behave more closely to a linear system at a certain range of frequencies. Experimental validation on an industrial motion stage demonstrates improved tracking performance, reduced higher-order harmonic content, and selective activation of the reset action in the intended frequency region.
\end{abstract}




\maketitle

\input{Sections/01Introduction.tex}
\input{Sections/02Preliminaries.tex}

\input{Sections/03Smart_Shaping_Filters.tex}

\input{Sections/04Simulation.tex}

\input{Sections/05Experiments}
\input{Sections/06Conclusion}
\appendix
\input{Sections/08Appendix_B}
\input{Sections/07Appendix_A}


\bibliographystyle{cas-model2-names}
\bibliography{TEXT}


\end{document}

%% file: Sections/01Introduction.tex
\section{Introduction}\label{Sec: Introduction}
Linear time-invariant (LTI) controllers, despite their widespread use, are inherently constrained by fundamental performance limitations, including the waterbed effect and Bode's gain-phase relationship \cite{freudenberg2000surveyBodeGain, waterbed2012}. Reset control has been introduced as a promising nonlinear control framework that has been shown to partially overcome such fundamental trade-offs \cite{clegg1958nonlinear,beker2004fundamental,guo2009frequency,ExpDemonst,ResetIntServo,secondOrderMarcel}.

The concept of the reset integrator \cite{clegg1958nonlinear} was extended to the first-order reset element (FORE) in \citep{HOROWITZFORE,KRISHNANFORE}. The FORE was later revisited in \citep{RevisitedZaccarian} from the perspectives of hybrid modeling and stability analysis. In \citep{guo2009frequency}, the generalized FORE (GFORE) was introduced, in which the post-reset state coefficient is allowed to take nonzero values. Among first-order reset elements, the GFORE has been widely used in reset control systems \citep{saikumar2019constant,hosseini2025AddOnFilterDesign}, as it provides additional tuning flexibility in terms of both the aggressiveness of the reset jumps and the frequency range over which reset action is activated. This motivates the focus of this study on this class of reset elements.

The frequency-domain analysis of reset elements using the sinusoidal-input describing function (SIDF) method was presented in \cite{guo2009frequency}. This framework was subsequently extended in \cite{saikumar2021loop,LukeLure} through the introduction of the higher-order SIDF (HOSIDF) \cite{NUIJ20061883}, along with a frequency-domain approach for relating open-loop characteristics to closed-loop dynamics in reset control systems.

While reset elements can provide favorable first-order harmonic behavior, as captured by the SIDF and exploited in the Constant-in-Gain, Lead-in-Phase (CgLp) element \cite{saikumar2019constant,hosseini2025AddOnFilterDesign}, the higher-order harmonics generated by reset action are often detrimental and may degrade closed-loop performance. To address this issue, several studies have employed the HOSIDF framework to shape the higher-order harmonics content while preserving the SIDF or maintaining it within a prescribed range \cite{cai2020optimal,karbasizadeh2022continuous,hosseini2025robust}. It was shown in \cite{cai2020optimal} that the influence of HOSIDFs can be mitigated without altering the SIDF by suitably arranging the loop components. Furthermore, \cite{karbasizadeh2022continuous,hosseini2025robust} apply a pre- and post-filtering approach, adapted from \cite{1161990,heertjes2009performance}, to reset control systems, thereby preserving the first-order harmonic response while enabling systematic shaping of the higher-order harmonic content of the error signal.

Although nonlinearities in the error signal can degrade closed-loop performance, nonlinearities in the reset signal, i.e., the reset-action command shown in Fig.~\ref{fig: reset block diagram}, are often more critical. In particular, undesired reset actions may lead to unreliable closed-loop behavior, especially when excessive reset instants occur. Moreover, the open-loop to closed-loop translation of HOSIDFs relies on the assumption that the reset action is primarily governed by the first-order harmonic of the reset signal \citep[Assumption~2]{saikumar2021loop}. However, since the reset signal inherently contains higher-order harmonics due to the nonlinear nature of the reset element, discrepancies inevitably arise between the assumed and actual reset instants. These observations motivated the work in \citep{BandpassNima}, where a method was proposed to shape the reset signal such that its nonlinear content is reduced, thereby improving the reliability of closed-loop performance prediction. The main idea was to confine the nonlinear behavior to the frequency range in which it is beneficial, while recovering linear dynamics in frequency regions where reset action is not required. This led to the introduction of shaping filters, for which the reset-action signal is no longer identical to the reset input signal, in contrast to conventional reset elements \cite{HallotFORE,RevisitedZaccarian,banos2012reset}.

Although the phase-shaping approach proposed in \cite{BandpassNima} provides an effective way to localize the reset-induced nonlinearity within a prescribed frequency band, it also increases the design and implementation complexity. Its performance depends on the accurate tuning and realization of the shaping filter, which may introduce sensitivity to modeling errors and practical implementation effects. While the method can effectively suppress undesired reset-induced harmonics, the added filtering dynamics may increase the controller order and reduce implementation simplicity. More importantly, the resulting HOSIDFs are mainly shifted downward in magnitude, while their overall frequency-dependent trend and low-frequency asymptotic behavior remain largely similar to those of the conventional reset element.
Motivated by these limitations, this paper revisits the GFORE with
shaping filters and develops a constructive design framework for
reducing undesired reset-induced nonlinearities. The novelty of the
proposed approach does not lie in the use of shaping filters alone, but
in deriving explicit algebraic conditions on the shaping-filter
coefficients that directly determine the low-frequency asymptotic
behavior of the HOSIDFs. In contrast to existing shaping-filter approaches that primarily attenuate
the HOSIDFs over selected frequency bands, the present work assigns the
low-frequency order of attenuation by imposing algebraic constraints on the
shaping-filter coefficients. This gives a direct synthesis rule for controlling
the asymptotic HOSIDF slope, rather than relying only on numerical phase
alignment or trial-and-error filter tuning.

The main contributions of this paper are summarized as follows. First,
a general shaping-filter structure is proposed for the GFORE, and
coefficient-level conditions are derived such that the low-frequency
attenuation slope of the HOSIDFs increases systematically with the
filter order. Second, the effect of the proposed shaping filters is
validated through HOSIDF analysis and a superposition-law assessment,
showing that the shaped reset element behaves closer to its base linear
system in frequency regions where reset action is not required. Third,
the proposed design is incorporated into a reset-based controller
and experimentally validated on an industrial wire-bonder motion stage.
The results demonstrate improved tracking performance, reduced
higher-order harmonic content, suppression of excessive reset actions,
and selective activation of the reset mechanism in the intended
frequency region.

The remainder of this paper is organized as follows. Section \ref{Sec: Preliminaries} presents preliminaries on reset elements with shaping filters and their frequency-domain representation. Section \ref{Sec: Smart SF} introduces the proposed shaping-filter design for reducing HOSIDFs. Section \ref{Sec: superposition law} validates the method through HOSIDF analysis and a superposition-law assessment. Section \ref{Sec: Experimental Example} presents the industrial wire-bonder case study and experimental validation. Finally, Section \ref{sec: conclusion} concludes the paper.

%% file: Sections/02Preliminaries.tex
\section{Preliminaries}\label{Sec: Preliminaries}
In this section, we introduce the reset element, including the shaping filter, in both the time and frequency domains.
\subsection{Reset element}
The reset element (depicted in Fig.~\ref{fig: reset block diagram}) is denoted by \( \mathcal{R} \) and defined as follows:  
\begin{equation}
\label{eq: reset state space}
\mathcal{R} := 
\begin{cases}
    \dot{x}_r(t) = A_r x_r(t) + B_r u_1(t), & \text{if } \left(x_r(t), e_r(t)\right) \notin \mathcal{F}, \\[5pt]
    x_r(t^+) = A_\rho x_r(t), & \text{if } \left(x_r(t), e_r(t)\right) \in \mathcal{F}, \\[5pt]
    u_r(t) = C_r x_r(t) + D_r u_1(t), &
\end{cases}
\end{equation}
where the state vector is \( x_r(t) \in \mathbb{R}^{n_r \times 1} \), and the post-reset state is \( x_r(t^{+}) \in \mathbb{R}^{n_r \times 1} \). The state-space matrices of the reset element are given by 
\( A_r \in \mathbb{R}^{n_r \times n_r} \), 
\( B_r \in \mathbb{R}^{n_r \times 1} \), 
\( C_r \in \mathbb{R}^{1 \times n_r} \), 
and \( D_r \in \mathbb{R} \). 
The reset matrix is denoted as 
\( A_\rho = \text{diag}(\gamma_1, \dots, \gamma_{n_r}) \), 
where \( -1 \leq \gamma_{i} \leq 1 \) for all \( i \in \{1,\dots,n_r\} \). 
The signals \( u_1(t) \in \mathbb{R} \) and \( u_r(t) \in \mathbb{R} \) represent the input and output of the reset element, respectively. Let \( C_s \) denote the LTI shaping filter
\begin{equation}
\label{eq: Cs dynamic ss}
C_s :=
\begin{cases}
    \dot{x}_z(t) = A_z x_z(t) + B_z u_1(t), \\[5pt]
    e_r(t) = C_z x_z(t) + D_z u_1(t),
\end{cases}
\end{equation}
with $x_z\in \mathbb{R}^{p\times 1}$, \( A_z \in \mathbb{R}^{p \times p} \), 
\( B_z \in \mathbb{R}^{p \times 1} \), 
\( C_z \in \mathbb{R}^{1 \times p} \), 
and \( D_z \in \mathbb{R} \). The signal \( e_r(t) \in \mathbb{R} \), referred to as the reset-triggering signal, determines when the state \( x_r \) is reset according to the reset surface \( \mathcal{F} \), defined as  

\begin{equation}
\label{reset surface}
\mathcal{F} := \{ e_r(t) = 0 \wedge (A_\rho - I) x_r(t) \neq 0 \}.
\end{equation}

\subsection{(Higher-Order) Sinusoidal Input Describing Functions}
The inherent nonlinearity of the reset element implies that its steady-state response to a sinusoidal input is non-sinusoidal. Consequently, its frequency response analysis often relies on approximations such as the SIDF method \cite{guo2009frequency}. However, because the SIDF accounts only for the first harmonic of the output while neglecting higher-order components, it may result in significant inaccuracies. To overcome this limitation, HOSIDFs were introduced for reset elements in \cite{saikumar2021loop}, and later extended to reset elements with shaping filters in \cite{ComplexKarbasi}.

\usetikzlibrary{arrows.meta, positioning, calc}

\newcommand{\figwidth}{0.3\columnwidth} 
\newcommand{\figfont}{\large}            

\begin{figure}
    \centering
    \resizebox{0.6\columnwidth}{!}{%
    \begin{tikzpicture}[
        >=Latex,
        font=\figfont,
        block/.style={
            draw,
            line width=1.2pt,
            minimum width=1.35cm,
            minimum height=1.15cm,
            inner sep=4pt
        },
        arr/.style={->, line width=1.2pt},
        darr/.style={->, dashed, line width=1.2pt}
    ]

        \coordinate (in)  at (0,0);
        \coordinate (tap) at (0.8,0);
        \coordinate (out) at (6.20,0);

        \node[block, minimum width=1.1cm, minimum height=1.0cm] (Cs) at (2.00,-1.25) {\LARGE{$C_s$}};
        \node[block, minimum width=1.4cm, minimum height=1.6cm] (R) at (4.70,0) {\huge{$\mathcal{R}$}};

        \draw[arr] (in) -- (R.west);
        \draw[arr] (R.east) -- (out);

        \draw[darr] (tap) -- (tap |- Cs.west) -- (Cs.west);

        \coordinate (erA) at ($(Cs.east)+(0.50,0)$); 
        \coordinate (erA2) at ($(Cs.east)+(0.50,0.8)$); 
        \coordinate (erA3) at ($(R.west)-(0.0,0.44)$); 
        \draw[darr] (Cs.east) -- (erA) -- (erA2) -- (erA3);
        \node[above] at (0.25,0.08) {$u_1$};
        \node[above] at (5.80,0.08) {$u_r$};
        \node[below] at (3.35,-0.50) {$e_r$};

    \end{tikzpicture}%
    }
    \caption{A reset element with the shaped reset signal. The resetting action is determined by $e_r(t)$.}\label{fig: reset block diagram}
\end{figure}
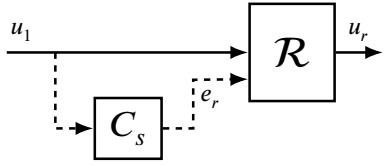

Thus, having the input of the reset element as $u_1(t)=\hat{u}_1\sin(\omega t)$, the output $u_r(t)$ can be described by the Fourier series:
\begin{equation}
    \label{eq: e to y}
u_r(t)=\sum_{n=1}^{\infty}\left|{H_{n}}(\omega)\right|\hat{u}_1\sin\left(n\omega t +\angle {H_{n}}(\omega)\right),
\end{equation}
with $n\in \mathbb{N}$, and ${H_{n}}(\omega)$ is SIDF and HOSIDFs of the reset element with shaping filter, which can be calculated as follows (see \cite[Section C]{ComplexKarbasi})
\begin{equation}
\label{eq: Hn}
H_{n}(\omega) =
\begin{cases}
C_r(A_r - j\omega I)^{-1}B_r\Theta_{\varphi}(\omega)\\
\qquad+ C_r(j\omega I - A_r)^{-1}B_r + D_r, & n=1, \\[6pt]
C_r(A_r - j\omega n I)^{-1}B_r\Theta_{\varphi}(\omega), & \text{odd } n \ge 2, \\[6pt]
0, & \text{even } n \ge 2,
\end{cases} \\[6pt]
\end{equation}
with
\begin{align}
\label{eq: HOSIDFs all functions}
\Theta_{\varphi}(\omega) &=
\frac{-2j\omega}{\pi}\,
\Omega(\omega)\Upsilon(\omega)\Lambda^{-1}(\omega) \nonumber\\[0pt]
\Upsilon(\omega)&=e^{j\varphi(\omega)}\Big(\omega I \cos\left(\varphi(\omega)\right) - A_r \sin\left(\varphi(\omega)\right)\Big)\nonumber\\
\Omega(\omega) &= \Delta(\omega) - \Delta(\omega)\Delta_{\rho}^{-1}(\omega) A_{\rho} \Delta(\omega)\nonumber \\[0pt]
\Lambda(\omega) &= \omega^{2}I + A_{r}^{2}\nonumber \\[0pt]
\Delta(\omega) &= I + e^{\tfrac{\pi}{\omega}A_{r}}\nonumber \\[0pt]
\Delta_{\rho}(\omega) &= I + A_{\rho} e^{\tfrac{\pi}{\omega}A_{r}}.\nonumber\\
& 
\end{align}
In the above equations, $\varphi(\omega) = \arg\left(C_s(j\omega)\right)$ denotes the phase shift in the reset signal relative to the case without a shaping filter. In the absence of a shaping filter, the reset occurs at the instants $t_k = \tfrac{k\pi}{\omega}$, where $\hat{u}_1 \sin(\omega t) = 0$. With the shaping filter included, however, the reset occurs at $t_k = \tfrac{k\pi - \varphi}{\omega}$, where $e_r(t)=\hat{u}_1 |C_s(j\omega)| \sin(\omega t + \varphi) = 0$.

A reset element follows its base linear system (BLS) dynamics if no reset happens $\left((x_r(t), e_r(t)) \notin \mathcal{F}\,\,\, \forall\,\,t \in \mathbb{R}_{\geq 0}\right)$. Thus, the transfer function of its BLS is defined as follows
\begin{equation}
\label{eq RCS bls}
R(j\omega) = C_r(j\omega I - A_r)^{-1}B_r + D_r.
\end{equation}

Having \(H_{ 1}(\omega)\) as the SIDF of the reset element, previous studies \cite{guo2009frequency, saikumar2021loop} have shown that although \(\left|H_{1}(\omega)\right|\) exhibits similar characteristics to its base linear system, it introduces less phase lag compared to it. This distinctive property of reset elements makes them a promising choice to overcome the Bode gain-phase limitation inherent in LTI systems. However, as discussed earlier, the SIDF considers only the first harmonic of the output while neglecting higher-order components, which can lead to significant inaccuracies. Therefore, it is of particular interest to reduce the magnitude of the HOSIDFs while preserving the favorable properties of the reset element (phase advantages). In the next section, we demonstrate how a carefully designed shaping filter, \(C_s(j\omega)\), can be used to attenuate the magnitude of HOSIDFs.

%% file: Sections/03Smart_Shaping_Filters.tex
\section{Shaping Filter Design for Reduction of HOSIDFs}
\label{Sec: Smart SF}
We consider a GFORE element defined as $A_r = -\omega_r$ with $\omega_r \in \mathbb{R}^+$, $B_r = 1$ and $C_r = \omega_r$. The direct feedthrough term \( D_r \) is set to zero in this illustration since it does not affect the HOSIDFs as expressed in \eqref{eq: Hn}. Fig. \ref{fig: GFORE wr 10 30 100} illustrates the first- and third-order SIDFs of the GFORE element for different corner frequencies \( \omega_r \).  It can be observed that although the magnitude of the third-order harmonic decreases at low frequencies, in many cases, a faster attenuation is desired (current attenuation is $40$ dB/dec). This reveals a design trade-off: a smaller \(\omega_r\) increases the useful phase-generation range of the corresponding CgLp element \cite{hosseini2025AddOnFilterDesign}, but it also amplifies the low-frequency higher-order harmonic content, thereby increasing the undesired nonlinear contribution in the same frequency range.
\begin{figure}
    \centering
    \includegraphics[width=1\columnwidth]{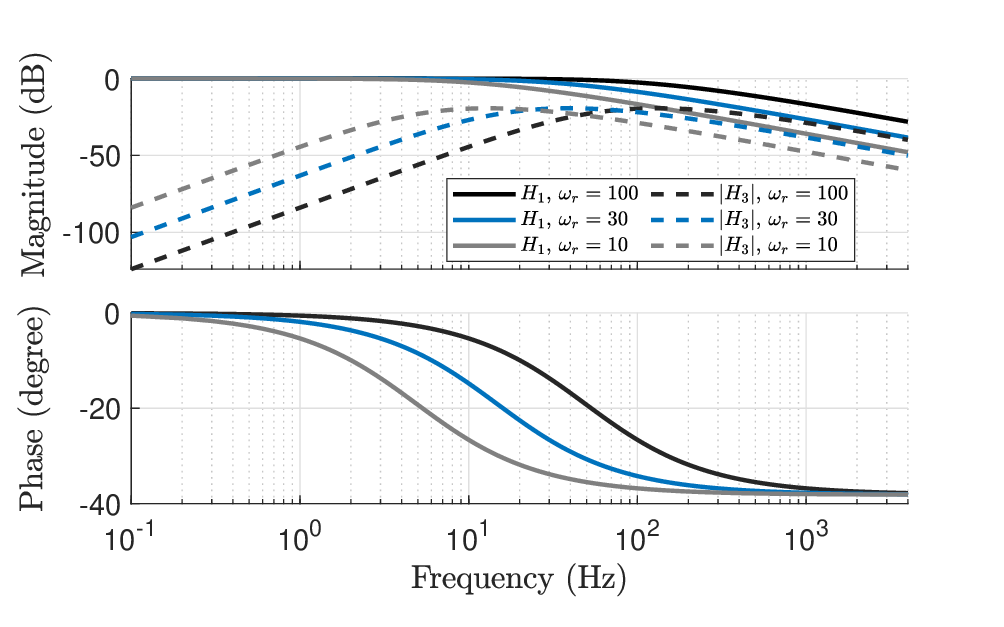}
    \caption{The first and third-order SIDF of the GFORE element with $\omega_r=10,30,100$ rad/sec.}
    \label{fig: GFORE wr 10 30 100}
\end{figure}

Thus, in this section, we propose a shaping filter structure that achieves a higher low-frequency attenuation slope of the HOSIDF magnitude as \( \omega \rightarrow 0 \), where this slope is defined as follows:
\begin{definition}
    \label{def: slop h0}
    The low-frequency slope (\textup{dB/dec}) of the magnitude response of the $H_n(\omega)$ is defined as
\begin{equation}
\label{eq: h0}
h_0 = \lim_{\omega \to 0} 
\frac{d\, 20\log_{10}|H_n(\omega)|}{d\,\log_{10}\omega},
\end{equation}
which quantifies the asymptotic rate of change of the magnitude at low frequencies.
\end{definition}

As depicted in Fig.~\ref{fig: GFORE wr 10 30 100}, the slope of all HOSIDFs of a GFORE element is \(40~\mathrm{dB/dec}\). Therefore, the only component that can alter this behavior is the shaping filter. In the following lemma, we explicitly express \( h_0 \) as a function of \( \Upsilon(\omega) \), which encapsulates the contribution of the shaping filter in the HOSIDF equations.  
\begin{lemma}
\label{lem: slop Upsilon}
Consider a first-order reset element in
\eqref{eq: reset state space}, with $A_r\in\mathbb{R}_{<0}$,
$B_r,C_r\in\mathbb{R}$, and scalar reset coefficient
$A_\rho\in(-1,1)$. Assume that
$C_rB_r(1-A_\rho)\neq 0$. For any fixed odd higher-order
harmonic $n\geq 3$, let $h_0$ denote the low-frequency slope
of $|H_n(\omega)|$. If
\begin{equation}
    \label{eq: var upsilon 0}
    \upsilon_0=\lim_{\omega \to 0^+}
    \frac{d\,\log_{10}|\Upsilon(\omega)|}
    {d\,\log_{10}\omega}
\end{equation}
exists, where $\Upsilon(\omega)$ is defined in
\eqref{eq: HOSIDFs all functions}, then
\begin{equation}
    \label{eq: Upsilon h0}
    h_0 = 20\left(\upsilon_0+1\right).
\end{equation}
\end{lemma}

\noindent
\textbf{Proof}:
For odd higher-order harmonics $n\geq 3$, \eqref{eq: Hn}
gives
\begin{equation}
    H_n(\omega)
    =
    C_r(A_r-jn\omega)^{-1}B_r\Theta_\varphi(\omega).
\end{equation}
Since $A_r\in\mathbb{R}_{<0}$, it follows that
$e^{\pi A_r/\omega}\to 0$ as $\omega\to 0^+$. Therefore,
from \eqref{eq: HOSIDFs all functions}, we have
\begin{equation}
    (A_r-jn\omega)^{-1} \to \frac{1}{A_r}, \qquad
    \Lambda^{-1}(\omega) \to \frac{1}{A_r^2},
\end{equation}
and $\Omega(\omega) \to 1-A_\rho$. Consequently,
\begin{equation}
    H_n(\omega)
    \sim
    -\frac{2jC_rB_r(1-A_\rho)}{\pi A_r^3}
    \Upsilon(\omega)\omega,
    \qquad \omega\to 0^+ .
\end{equation}
Hence,
\begin{equation}
    |H_n(\omega)|
    \sim
    K|\Upsilon(\omega)|\omega,
    \qquad
    K=
    \left|
    \frac{2C_rB_r(1-A_\rho)}{\pi A_r^3}
    \right| .
    \label{eq: hn asymptotic}
\end{equation}
Since $K>0$ is constant, it does not affect the logarithmic
slope. Therefore,
\begin{align}
&\hspace{-1.5em} h_0
=
\lim_{\omega\to 0^+}
\frac{d\,20\log_{10}|H_n(\omega)|}
{d\,\log_{10}\omega} \nonumber \\
&\hspace{-1.5em}=
20\lim_{\omega\to 0^+}
\frac{d\,\log_{10}\left(K|\Upsilon(\omega)|\omega\right)}
{d\,\log_{10}\omega} \nonumber \\
&\hspace{-1.5em}=
20\left(
\lim_{\omega\to 0^+}
\frac{d\,\log_{10}|\Upsilon(\omega)|}
{d\,\log_{10}\omega}
+
\lim_{\omega\to 0^+}
\frac{d\,\log_{10}(K\omega)}
{d\,\log_{10}\omega}
\right).
\end{align}
Using
\begin{equation}
    \lim_{\omega\to 0^+}
    \frac{d\,\log_{10}(K\omega)}
    {d\,\log_{10}\omega}=1
\end{equation}
and the definition of $\upsilon_0$ in \eqref{eq: var upsilon 0},
we obtain
\begin{equation}
    h_0 = 20\left(\upsilon_0+1\right).
\end{equation}
This completes the proof. \qed

Regarding Lemma~\ref{lem: slop Upsilon}, the magnitude \( |\Upsilon(\omega)| \) determines the slope of \( |H_n(\omega)| \). In the typical case without a shaping filter, where \( C_s(j\omega) = 1 \) (or equivalently \( \varphi(\omega) = 0 \)), we obtain \( |\Upsilon(\omega)| = \omega \) which results in $\upsilon_0=1$, and consequently \( h_0 = 40~\mathrm{dB/dec} \). In the following, we introduce a theorem that provides a shaping filter design and parameter selection guideline to modify \( |\Upsilon(\omega)| \) such that \( h_0 \) increases, thereby reducing the influence of the HOSIDFs.

\begin{theorem}
\label{theorem: shaping filter}
Consider a first-order reset element defined by
\[
A_r = -\omega_r,\quad 
B_r\in \mathbb{R},\quad 
C_r \in \mathbb{R},\quad 
D_r \in \mathbb{R},
\]
with $\omega_r \in \mathbb{R}_{>0}$, the reset coefficient satisfy
\(A_\rho\in(-1,1)\), and assume \(C_rB_r(1-A_\rho)\neq0\). Let the shaping filter be a bi-proper Hurwitz transfer function of order $p$ of the form
\begin{equation}
    \label{eq: Cs}
    C_s(s) = 
    \frac{\sum_{k=0}^{p} a_k\,\omega_r^{-k} s^{k}}
         {\sum_{k=0}^{p} b_k\,\omega_r^{-k} s^{k}},
\end{equation}
where $p \in \mathbb{N}$, $a_k,b_k\in\mathbb{R}$,
$a_0b_0\neq 0$, and $a_pb_p\neq 0$. Assume that no pole-zero
cancellation reduces the effective order of $C_s(s)$. If, for each
$m = 1,\ldots,p$, the coefficients satisfy
\begin{equation}
    \label{eq: Eq_m Th}
    \mathrm{Eq}_m: \quad
    \mathop{\sum_{l=0}^{p} \sum_{k=0}^{p}}_{k+l=2m-2}
    a_k b_l (-1)^{l}
    +
    \mathop{\sum_{l=0}^{p} \sum_{k=0}^{p}}_{k+l=2m-1}
    a_k b_l (-1)^{l}
    = 0,
\end{equation}
then the low-frequency slope of the HOSIDFs satisfies
\begin{equation}
    \upsilon_0 = 2p+1,
\end{equation}
and therefore
\begin{equation}
    h_0 = 20(\upsilon_0+1)=40(p+1)\ \mathrm{dB/dec}.
\end{equation}
\end{theorem}

\textbf{Proof:} See Appendix~\ref{App: pf of theorem}.

Theorem \ref{theorem: shaping filter} provides a general structure for the shaping-filter design together with its parameter-tuning conditions. As shown, it yields $p$ equations for $2(p+1)$ unknown parameters in the shaping filter. Consequently, the solution for achieving $\upsilon_0 = 2p+1$ is not unique, offering flexibility in designing an appropriate shaping filter for various applications.

In the following, we present an example illustrating how a shaping filter can be constructed using Theorem \ref{theorem: shaping filter}. Considering a second-order shaping filter ($p=2$), from \eqref{eq: Cs} we obtain
\begin{equation}
\label{eq: second-order}
C_s(s) = \frac{a_0 + \frac{a_1}{\omega_r} s + \frac{a_2}{\omega_r^2} s^2}{b_0 + \frac{b_1}{\omega_r} s + \frac{b_2}{\omega_r^2} s^2}.
\end{equation}
Applying Theorem \ref{theorem: shaping filter}, the resulting conditions become
\begin{align}
\label{eq: Eq2 1}
\mathrm{Eq}_1: &\quad a_0 b_0 + a_1 b_0 - a_0 b_1 = 0, \nonumber\\
\mathrm{Eq}_2: &\quad -a_1 b_1 + a_0 b_2 + a_2 b_0 + a_1 b_2 - a_2 b_1 = 0.
\end{align}

Without loss of generality, we adopt the standard transfer-function normalization $a_0 = 1$ and $b_0 = 1$. The parameters $b_1$ and $b_2$ can then be selected according to the desired shaping-filter characteristics (e.g., cutoff frequency, presence of resonance) while ensuring that the poles remain in the left-half plane. The coefficients $a_1$ and $a_2$ follow directly from solving $\mathrm{Eq}_1$ and $\mathrm{Eq}_2$. The resulting shaping filter guarantees a HOSIDF slope of $h_0 = 20(2p + 1 + 1) = 120 \,\text{dB/dec}$. The same procedure can be applied to shaping filters of any order. To facilitate practical implementation, Appendix~\ref{app:practical_workflow} summarizes the design workflow for selecting the shaping-filter order, computing the coefficients, and verifying the resulting reset behavior.

In the next section, we show how the proposed shaping filters modify the HOSIDFs of a GFORE element and demonstrate their time-domain effectiveness.

%% file: Sections/04Simulation.tex
\section{HOSIDF Reduction Validation and Superposition Law Assessment}
\label{Sec: superposition law}
In this section, we first consider three shaping filters and illustrate their effect on the third-order harmonic of a GFORE element. We then investigate the superposition law for the GFORE element both with and without the inclusion of the designed shaping filters.

Let $C_{s_1}(s)$, $C_{s_2}(s)$, and $C_{s_3}(s)$ denote the first-, second-, and third-order shaping filters designed using Theorem~\ref{theorem: shaping filter}, respectively. Using \eqref{eq: Hn}, the quantities $H_1(\omega)$ and $|H_3(\omega)|$ for all four cases are plotted in Fig.~\ref{fig: GFORE HOSIDFs SF}. The GFORE parameters are identical across all cases; only the shaping filter varies (see Table \ref{tab: SF parameters}). For a fair comparison, all designs are constructed to exhibit the same phase at a nominal bandwidth frequency ($100$~Hz), and they also show nearly identical first-order harmonic magnitudes at $H_1(\omega)$.

\begin{figure}
    \centering
    \includegraphics[width=1\columnwidth]{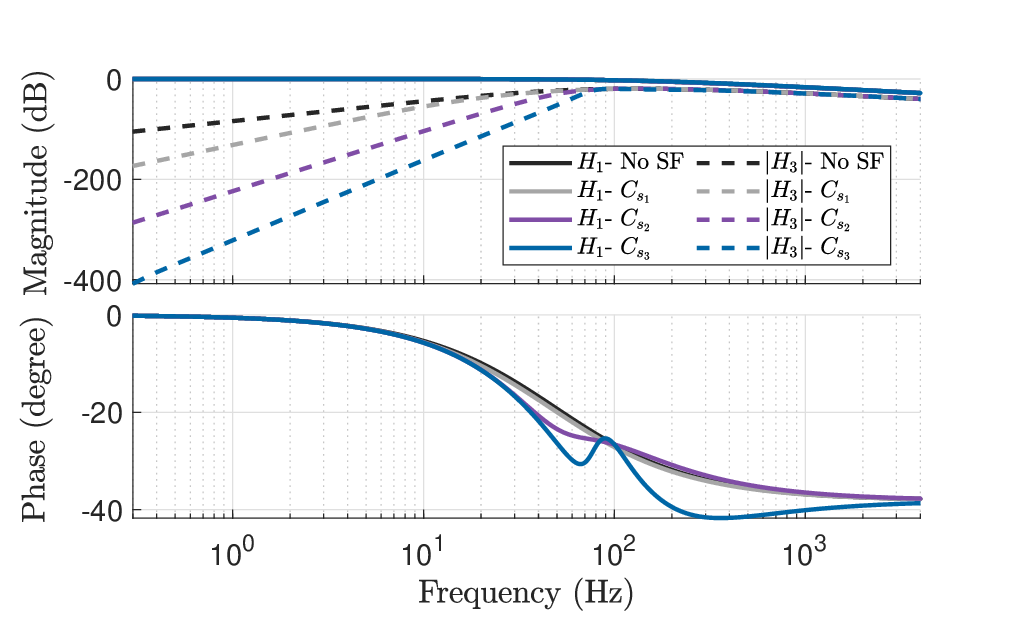}
    \caption{The first and third-order SIDF of the GFORE elements in cases: No shaping filter ($p=0,\,C_s=1$), first-order ($p=1$), second-order ($p=2$), and third-order ($p=3$) shaping filters}
    \label{fig: GFORE HOSIDFs SF}
\end{figure}

\begin{table}
\caption{Parameters of the designed shaping filters for the GFORE element with $\omega_r = 2\pi \times 100$ rad/s and $A_\rho = 0$.
}
\label{tab: SF parameters}
\resizebox{1\columnwidth}{!}{%
\begin{tabular}{ccccccccc}
\hline
          & \large{$a_0$} &\large{$a_1$} & \large{$a_2$} & \large{$a_3$}& \large{$b_0$}& \large{$b_1$} &\large{$b_2$}           &\large{$b_3$}    \\ \hline
\large{No SF}       & \large{$1$}        &\large{$0$}& \large{$0$}                &\large{$0$}&\large{$1$}& $0$ &$0$         & \large$0$ \\ 
\large{$C_{s_1}$}       & \large{$1$}        &\large{$6$}& \large{$0$}                &\large{$0$}&\large{$1$}&\large{$7$} &\large{$0$}         & \large{$0$} \\ 
\large{$C_{s_2}$}       & \large{$1$}        &\large{$1.9046$}& \large{$2.7951$}                &\large{$0$}&\large{$1$}&\large{$2.9046$} &\large{$3.7374$}         & \large{$0$} \\
\large{$C_{s_3}$}       & \large{$1$}        &\large{$1$}& \large{$1.9021$}                &\large{$0.7906$}&\large{$1$}&\large{$2$} &\large{$2.3940$}         & \large{$1.6763$} \\
\hline
\end{tabular}
}
\end{table}

Examining $H_3(\omega)$ confirms the result of Theorem~\ref{theorem: shaping filter}: the HOSIDF slope scales with the shaping-filter order as $h_0 = 40(p+1)$~dB/dec. Since $p=0$ for the case without shaping filter, we obtain $h_0 = 40$~dB/dec. As depicted in Fig.~\ref{fig: GFORE HOSIDFs SF}, the higher-order filters yield slopes of $h_0 = 80$~dB/dec, $120$~dB/dec, and $160$~dB/dec for $p=1$, $p=2$, and $p=3$, respectively.

This observation is significant because the reset element is intended to operate predominantly around the bandwidth frequency, where its nonlinearity provides the desired phase advantage. At frequencies outside this region, we aim to suppress nonlinear effects to avoid introducing undesired harmonic distortions. The proposed shaping-filter design effectively reduces the HOSIDFs over the frequency range in which activation of the reset element is undesired, i.e., where the controller is intended to behave approximately linearly.

In this regard, to also validate the reduction of nonlinearity in the time domain, we examine the superposition property for the GFORE element, with and without the proposed shaping filters, through an illustrative example. With reference to Fig.~\ref{fig: reset block diagram}, we apply a multi-sine input $u_1(t)$ as
\begin{equation}
    u_{1}(t)=\sum_{i=1}^{6} u_{1_i}(t), \qquad 
    u_{1_i}(t)=\sin\bigl(2\pi f_i t + \phi_i\bigr),
\end{equation}
where each phase $\phi_i$ is independently drawn from a uniform distribution
$\phi_i \sim \mathcal{U}[0,2\pi)$. 
The corresponding output of the reset element is denoted by $u_r(t)$.
In addition, we excite the reset element individually with each $u_{1_i}(t)$, yielding outputs
$u_{r_i}(t)$. We then construct
\begin{equation}
    \hat{u}_r(t)=\sum_{i=1}^{6} u_{r_i}(t).
\end{equation}

For any LTI system, it necessarily holds that $u_r(t)=\hat{u}_r(t)$. However, due to the inherent nonlinearity of reset elements, this equality does not generally hold. Our interest, therefore, is in assessing how close $u_r(t)$ and $\hat{u}_r(t)$ become when the proposed shaping filters are introduced. To quantify this, we evaluate the normalized maximum superposition error
\begin{equation}
\label{eq: superposition law eq}
e_\mathrm{spl} = 
\frac{\max \left| u_r(t) - \hat{u}_r(t) \right|}
{\max \left| u_r(t) \right| }
\times 100.
\end{equation}

Using the frequency set $f_i = [1,\,5,\,10,\,15,\,20,\,50]~\mathrm{Hz}$, we compute $u_r(t)$ and $\hat{u}_r(t)$ for all configurations shown in Fig.~\ref{fig: GFORE HOSIDFs SF}. The time-domain differences $u_r(t) - \hat{u}_r(t)$ are presented in Fig.~\ref{fig: ur hat}. These results indicate that, as the order of the shaping filter increases, the discrepancy between $u_r(t)$ and $\hat{u}_r(t)$ decreases. This observation demonstrates the effectiveness of the proposed method in attenuating reset actions at low frequencies. Furthermore, Fig.~\ref{fig: e spl} reports the corresponding values of $e_\mathrm{spl}$, providing a quantitative measure of the improvement achieved by the shaping filters and enabling direct comparison across filter orders. In the next section, we demonstrate how the new shaping-filter design improves the performance of the closed-loop control system of an industrial motion stage.
\begin{figure}
    \centering
    \includegraphics[width=1\columnwidth]{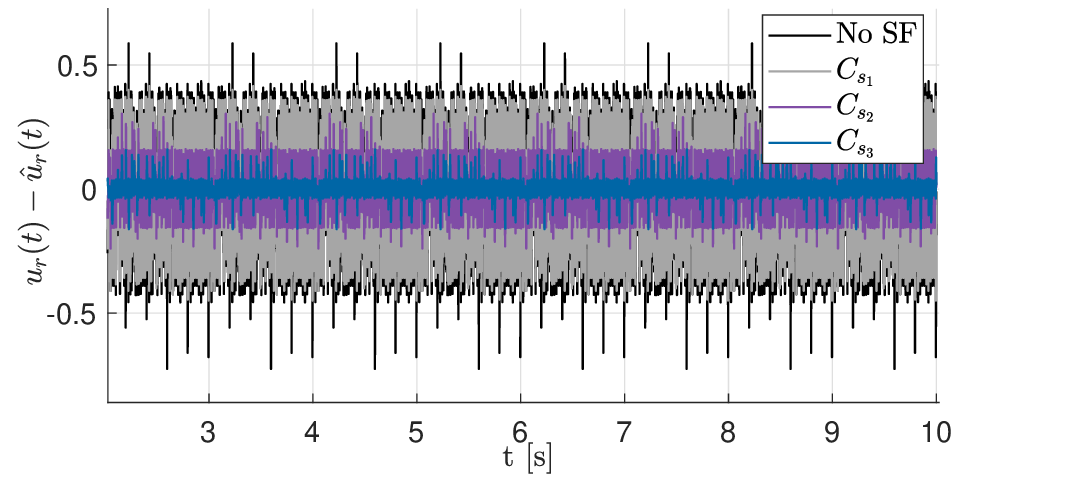}
    \caption{The value of $u_r(t) - \hat{u}_r(t)$ under different cases.}
    \label{fig: ur hat}
\end{figure}

\begin{figure}
    \centering
    \includegraphics[width=1\columnwidth]{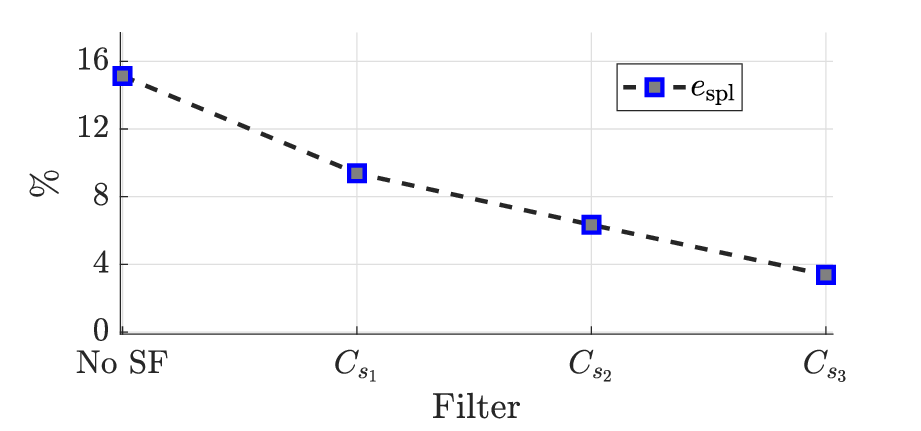}
    \caption{The percentage of $e_{\mathrm{spl}}$ for four cases.}
    \label{fig: e spl}
\end{figure}

%% file: Sections/05Experiments.tex
\section{Design and Experimental Results}\label{Sec: Experimental Example}
This section first introduces the case study and then presents the controller design and experimental results.
\subsection{Case Study}\label{subsec: case study}
A wire bonder is a key component in semiconductor manufacturing, used to form electrical interconnections between an integrated circuit and the terminals of its package. The ASMPT wire bonder considered in this work is shown in Fig. \ref{fig: Wire-bonder}a. Its isolated motion stage, depicted in Fig. \ref{fig: Wire-bonder}b, provides three degrees of freedom (DoF) along the X-, Y-, and Z-axes, with each stage actuating motion in the corresponding direction.

The analysis in this paper is restricted to the X-stage. The frequency response functions (FRFs) from the actuator forces $F_\mathrm{x}$, $F_\mathrm{y}$, and $F_\mathrm{z}$ to the X-direction displacement $D_\mathrm{x}$ are presented in Fig. \ref{Fig: FRF AB383}. The results show that the X–Y and X–Z cross-couplings are negligible, as their FRF magnitudes remain approximately 40 dB, corresponding to two orders of magnitude, below that of the direct transfer from $F_\mathrm{x}$ to $D_\mathrm{x}$ (around the targeted cross-over frequency). This indicates that each degree of freedom is actuated independently, such that the dynamics along the X-, Y-, and Z-axes can be accurately approximated as decoupled single-input–single-output (SISO) systems. For confidentiality, the frequency axis in Fig.~\ref{Fig: FRF AB383},
as well as in all subsequent experimental and frequency-domain results,
is reported in normalized form. Let $f_{\mathrm{phys}}$ denote the
physical frequency and let $\lambda_{\mathrm{scale}}>0$ be the confidential
scaling factor. The normalized frequency is defined as $\bar{f}=\frac{f_{\mathrm{phys}}}{\lambda_{\mathrm{scale}}}$. Equivalently, for angular frequencies, $\bar{\omega}
    =
    \frac{\omega_{\mathrm{phys}}}{\lambda_{\mathrm{scale}}}$. Unless otherwise stated, all frequency values and frequency axes reported in the remainder of this study correspond to these normalized quantities.

An optimal linear controller is synthesized for the system under consideration to improve low-frequency performance while enforcing the robustness constraint that the maximum sensitivity peak does not exceed 6 dB. The resulting automatically tuned controller, denoted by $C_\text{L}(\omega)$, is composed primarily of a PID compensator together with notch and inverse-notch filters at higher frequencies. For confidentiality reasons, the detailed controller structure and parameter values are not reported. This controller is optimized to maximize the corner frequency of the integral action, represented by a PI term of the form $\mathrm{PI}=1+\omega_i/s$, for which the achieved value is $\omega_i=2\pi\times 3.87\times10^{-3}\,$rad/sec. To further improve performance, the corner frequency must be increased without compromising robustness. To this end, reset controllers are used in the subsequent analysis.

\begin{figure}
\centering
\includegraphics[width=0.9\columnwidth]{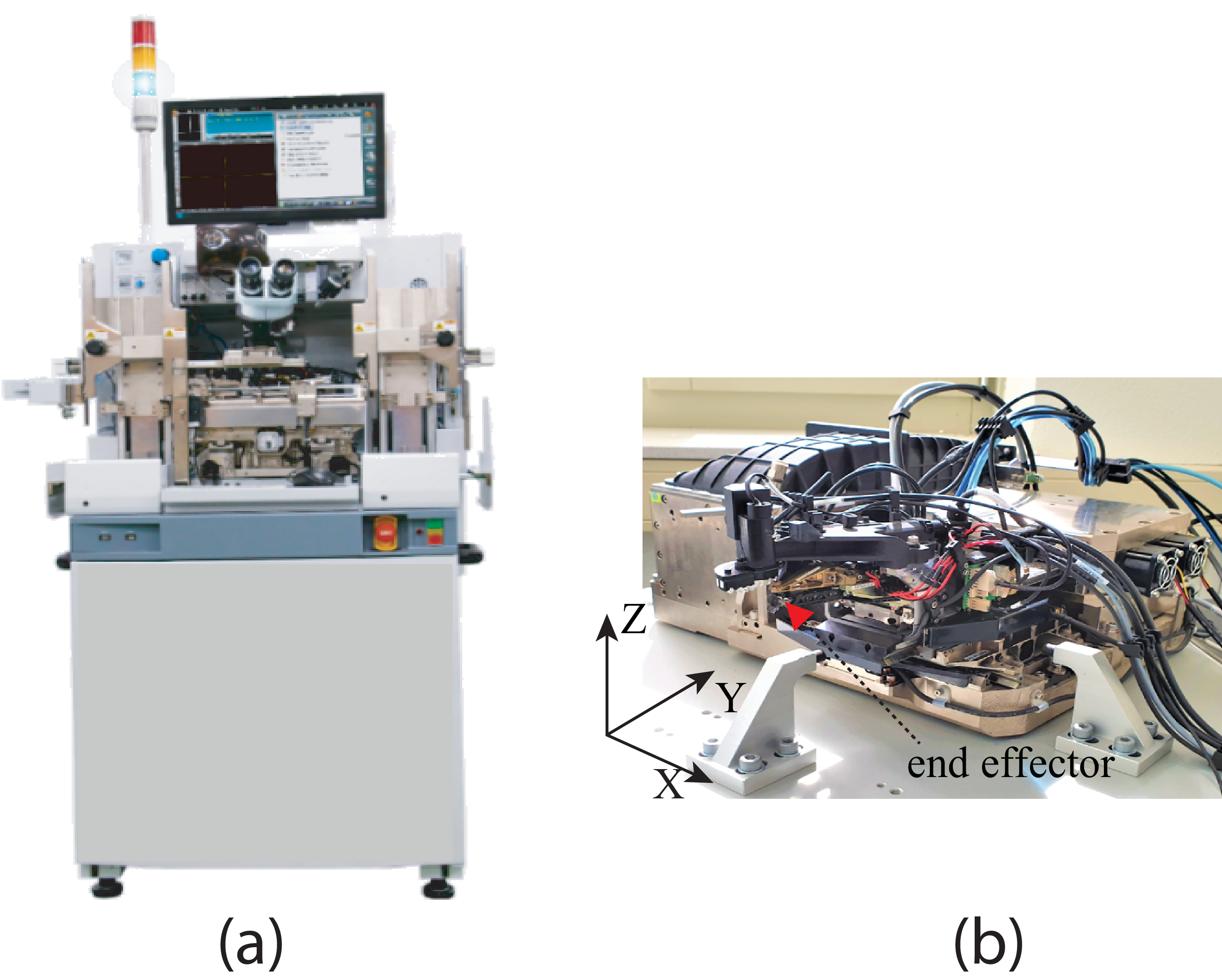}
\caption{(a) An Industrial wire bonder. (b) Isolated XYZ-motion platform of the wire bonder.}
\label{fig: Wire-bonder}
\end{figure}

\begin{figure}
\centering
\includegraphics[width=1\columnwidth]{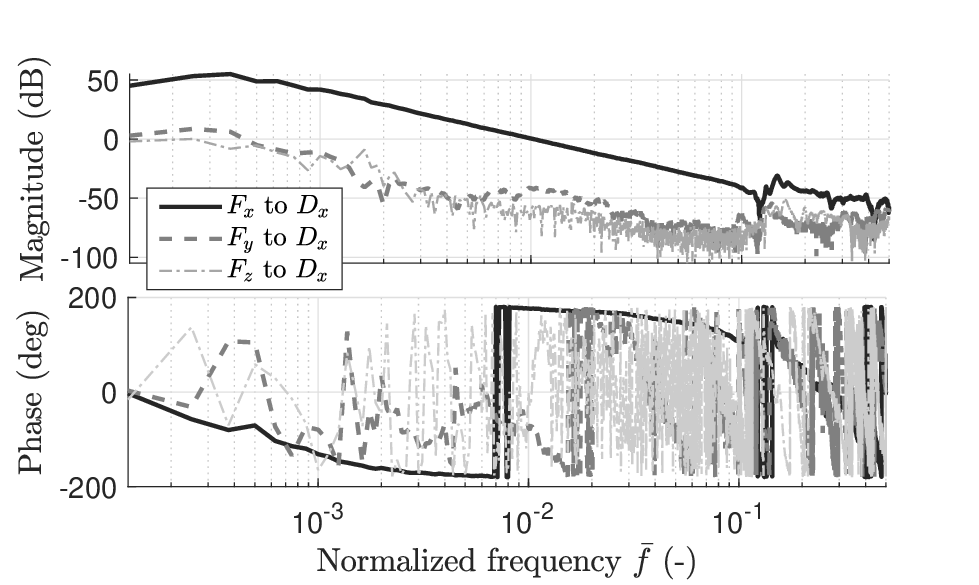}
\caption{FRF of the X-stage of the physical wire bonder, illustrating the mapping of actuator forces in the X-stage ($F_x$), Y-stage ($F_y$), and Z-stage ($F_z$) to the displacement measured by the X-stage encoder ($D_x$).}
\label{Fig: FRF AB383}
\end{figure}

\subsection{Reset control design}
In this section, CgLp element is used to increase the corner frequency of the PI controller. More specifically, a reset-based CgLp filter is used to compensate for the phase loss at the bandwidth resulting from the additional phase lag associated with an increase in $\omega_i$. The CgLp filter provides broadband phase lead in the frequency region of interest, particularly around the bandwidth, while maintaining approximately unity gain (0 dB). Since the detailed structure of the CgLp element is beyond the scope of this paper, the reader is referred to \citep[Definition 1]{hosseini2025AddOnFilterDesign} and \citep[Section V]{hosseini2025AddOnFilterDesign} for a comprehensive description.

\usetikzlibrary {arrows.meta}
\tikzstyle{block} = [draw,thick, fill=white, rectangle, minimum height=2em, minimum width=2.5em, anchor=center]
\tikzstyle{sum} = [draw, fill=white, circle, minimum height=0.6em, minimum width=0.6em, anchor=center, inner sep=0pt]
\usetikzlibrary {arrows.meta}
\tikzstyle{block} = [draw,thick, fill=white, rectangle, minimum height=1.75em, minimum width=2.1875em, anchor=center]
\tikzstyle{block2} = [draw,thick, fill=white, rectangle, minimum height=0.8em, minimum width=1.0em, anchor=center]
\tikzstyle{sum} = [draw, fill=white, circle, minimum height=0.6em, minimum width=0.6em, anchor=center, inner sep=0pt]
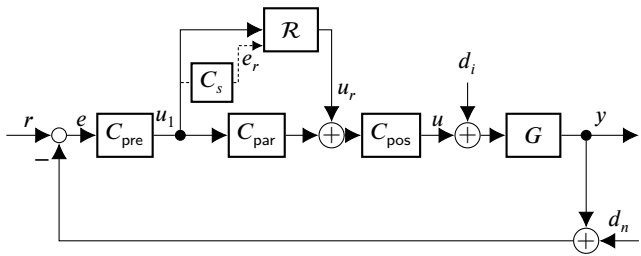
\begin{figure}
	\centering
	\begin{scaletikzpicturetowidth}{\columnwidth}
		\begin{tikzpicture}[scale=\tikzscale]
			\node[coordinate](input) at (0,0) {};
			\node[sum] (sum1) at (1,0) {};
			\node[sum] (sum3) at (8.75,0) {+};
			\node[sum] (sum4) at (11,-2) {+};
			\node[sum, fill=black, minimum size=0.4em] (dot2) at (11,0) {};	
            \node[sum, fill=black, minimum size=0.0em] (dotCs) at (3.3,1.0) {};
            \node[sum, fill=black, minimum size=0.0em] (dotImag) at (4.4,1.7) {};
            \node[sum, fill=black, minimum size=0.0em] (dotImag2) at (4.88,1.7) {};
			\node[sum] (sumN) at (6.17,0) {+};
			\node[block] (lead) at (5.4,2.0) {$\mathcal{R}$};
            \node[block2] (Cs) at (3.9,1.0) {$C_s$};
\node[block] (C_par) at (4.75,0) {$C_\text{par}$};
            
			\node[block] (controller) at (7.3,0) {$C_\text{pos}$};
			\node[block] (fo-higs) at (2.25,0) {$C_\text{pre}$};
			\node[block] (system) at (10,0) {$G$};

            \node[sum, fill=black, minimum size=0.4em] (Dot_1) at (3.3,0) {};
			\node[coordinate](output) at (12,0) {};
			\node[coordinate](di-input) at (8.75,1) {};
			\node[coordinate](n-input) at (12,-2) {};
			\draw[arrows = {-Latex[width=6pt, length=6pt]}] (input)  -- node[above]{$r$} (sum1);
			\draw[arrows = {-Latex[width=6pt, length=6pt]}] (di-input)node[above]{$d_i$}  --  (sum3);
			\draw[arrows = {-Latex[width=6pt, length=6pt]}] (n-input)  -- node[above]{$d_n$} (sum4);
			\draw[arrows = {-Latex[width=6pt, length=6pt]}] (sum1)   --node[above]{$e$}  (fo-higs);
			\draw[ = {-Latex[width=6pt, length=6pt]}] (fo-higs) --node [above]{$\,u_1$} (Dot_1);
            \draw[ = {-Latex[width=6pt, length=6pt]}, dash pattern=on 1.1pt off 0.9pt] (dotCs) --node [above]{} (Cs);
            \draw[  = {-Latex[width=6pt, length=6pt]}, dash pattern=on 1.1pt off 0.9pt] (Cs) -|node [above]{} (dotImag);
            \draw[arrows={-Latex[width=4pt, length=4pt]}, dash pattern=on 1.1pt off 0.9pt] 
  (dotImag) -- node [below]{$e_r$}(dotImag2);
			\draw[arrows = {-Latex[width=6pt, length=6pt]}] (lead)  -| (sumN) node[above, xshift=+6pt,yshift=+10pt] {$u_r$};

            \draw[arrows = {-Latex[width=6pt, length=6pt]}] (C_par)  --node[above] {} (sumN);
            
			\draw[arrows = {-Latex[width=6pt, length=6pt]}] (controller)  --node[above] {$u$} (sum3);
			\draw[arrows = {-Latex[width=6pt, length=6pt]}] (sum3) -- (system);
			\draw[arrows = {-Latex[width=6pt, length=6pt]}] (system)  -- node[above]{$y$} (output);
			\draw[arrows = {-Latex[width=6pt, length=6pt]}] (dot2) -- (sum4);

\draw[arrows = {-Latex[width=6pt, length=6pt]}] (sumN) -- (controller);
            
			\draw[arrows = {-Latex[width=6pt, length=6pt]}] (sum4) -| node[pos=0.9,left]{$-$} (sum1);

            \draw[arrows = {-Latex[width=6pt, length=6pt]}] (Dot_1) |-node[pos=0.85,left]{} (lead);

            \draw[arrows = {-Latex[width=6pt, length=6pt]}] (Dot_1) |-node[pos=0.85,left]{} (C_par);
			
		\end{tikzpicture}
	\end{scaletikzpicturetowidth}
	\caption{Block diagram of the closed-loop reset control system.}
	\label{Fig: Block diagram CL}	
\end{figure}

To study and analyze the frequency response of a closed-loop reset control system, we use the results from \cite{LukeLure} and \cite{hosseini2025reliability}. Having the closed-loop system in Fig. \ref{Fig: Block diagram CL} subject to input $r(t)=r_{o}\sin(\omega t)$ (with $r_o=1$ since the reset element does not exhibit amplitude-dependent nonlinearities), we present the pseudo-sensitivity for the error signal, as \citep[Section~II.D]{LukeLure}:
\begin{equation}
    \label{eq: S infty}
    \left|S_{\infty}(\omega)\right|
    = \frac{\max\limits_{0 \le t < 2\pi/\omega} |e(t,\omega)|}{r_o},
\end{equation}  
where
\begin{equation}
    \label{eq ess}
        e_\text{}(t,\omega)=\sum_{n=1}^{\infty}|S_{r,e}^{n}(\omega)|\sin{(n\omega t+\angle{S_{r,e}^{n}(\omega)})}.
\end{equation}
$S_{r,e}^{n}(\omega)$ are the higher-order sensitivity functions  ($\text{n}^\text{th}$-order HOSIDF, from $r$ to $e$) as \citep[equations (32)-(34)]{LukeLure}:
\begin{equation}
\label{eq Sn}
S_{r,e}^{n}(\omega) =
\begin{cases}
\displaystyle \frac{1}{1 + \mathcal{L}_1(\omega)}, \hfill \text{for } n = 1, \\[10pt]
\displaystyle -\mathcal{L}_n(\omega) S_\mathrm{bl}(nj\omega) 
\left(|S_{r,e}^{1}(\omega)| e^{jn\angle S_{r,e}^{1}(\omega)}\right), \\
\hfill \text{for odd } n \geq 2, \\[5pt]
0, \hfill \text{for even } n \geq 2,
\end{cases}
\end{equation}
 where $S_\mathrm{bl}(nj\omega)=\frac{1}{1+L_\mathrm{bl}(nj\omega)}$, and
 \[L_\mathrm{bl}(j\omega)=G(j\omega)C_\text{pos}(j\omega)[C_\text{par}(j\omega)+R_\text{bl}(j\omega)]C_\text{pre}(j\omega)\]
 is the base linear transfer function of the open-loop. \( \mathcal{L}_n(j\omega) \) $\forall n\neq1$ is given by (see \citep[equation (34)]{LukeLure}):  
\begin{equation}
    \label{eq: Ln open loop}
\mathcal{L}_{n}(\omega)=G(nj\omega)C_\text{pos}(nj\omega)H_{\varphi n}(\omega)C_\text{pre}(j\omega)e^{j(n-1)\angle C_\text{pre}(j\omega)},
\end{equation}
where for $n=1$ we have \cite[equation (30)]{LukeLure}

\begin{equation}
        \label{eq: L1 open loop}
    \mathcal{L}_{1} (\omega)=G(j\omega)C_\mathrm{pos}(j\omega)[H_{\varphi 1}(\omega)+C_\mathrm{par}(j\omega)]C_\mathrm{pre}(j\omega).
\end{equation}

\begin{remark}
\label{rem: multi reset}
Please note that the existence of the closed-loop frequency response in reset control systems relies on the assumption that the reset instants $t_k$ occur $\pi/\omega$ apart and are determined by the first-order harmonic of the reset signal $e_r(t)_{n=1}$ (see \citep[Assumption~2]{saikumar2021loop}).
\end{remark}

Having reset instants spaced by $\pi/\omega$ implies that if the reset signal $e_r(t)$ exhibits multiple zero-crossings within one cycle, there will inevitably be more than two reset instants within one $2\pi/\omega$ cycle. In such a case, the closed-loop frequency response of the reset control system is no longer valid. Let $N_r \in \mathbb{N}$ denote the number of resets occurring within one period of the input, defined as the number of distinct solutions of the reset signal equation $e_r(t,\omega)=0$ over the interval $t \in [0, 2\pi/\omega)$. Formally,
\begin{equation} \label{eq: nr}
N_r(\omega) \;:=\; \#\Bigl\{\, t_k\in\bigl[0,\tfrac{2\pi}{\omega}\bigr) : e_r(t_k,\omega)=0 \,\Bigr\}, \,\,\,\,\, N_r \geq 2.
\end{equation}
The reset signal $e_r(t)$ is calculated as \cite[Section 2.3]{hosseini2025reliability}
\begin{equation}
    \label{eq:er}
        e_{r}(t,\omega)=
        \sum_{n=1}^{\infty}\left|S_{r,e_r}^{n}(\omega)\right|\sin{\left(n\omega t+\angle{S_{r,e_r}^{n}(\omega)}\right)},
\end{equation}
 where
 \begin{equation}
\label{eq: S1-u1}
    S_{r,e_r}^{1}(\omega) = \frac{C_{\text{pre}}(j\omega)C_s(j\omega)}{1 + \mathcal{L}_1(\omega)},
\end{equation}
and
\begin{align}
& S_{r,e_r}^{n}(\omega)
\notag\\
& =
\frac{
-G(nj\omega)C_{\text{pos}}(nj\omega)H_{\varphi n}(\omega)
C_{\text{pre}}(nj\omega)C_s(nj\omega)
}
{1 + L_{\text{bl}}(nj\omega)}
\notag\\
& {}\cdot
S_{r,e_r}^{1}(\omega)
e^{j(n-1)\angle S_{r,e_r}^{1}(\omega)},
\quad \forall n \neq 1 .
\label{eq:Sn-u1}
\end{align}

Using the closed-loop frequency-domain representations developed for reset control systems with shaping filters, the following sections first design a reset controller to improve upon the linear controller \(C_\mathrm{L}\). Subsequently, the performance is further enhanced by incorporating a shaping filter into the reset controller design.

\subsubsection{Reset control without shaping filter}

We first consider the reset controller \(C_{\mathcal{R}_1}\), which does not include a shaping filter. This controller is designed to maximize the achievable integral corner frequency \(\omega_i\). To compensate for the associated phase loss, a CgLp filter is incorporated into the control design. Starting from the nominal design, \(\omega_i\) was increased incrementally, and the corresponding phase compensation provided by the CgLp element was adjusted accordingly. This procedure was continued until the induced nonlinearity became excessive, leading to multiple zero-crossings \((N_r > 2)\). Following this design approach, the maximum achievable value was found to be
\(\omega_i = 2\pi \times 7 \times 10^{-3}\ \text{rad/s}.\)
The corresponding CgLp parameters are listed in Table \ref{tab: controller parameters}. The sensitivity functions obtained with \(C_{\mathcal{R}_1}\) and the linear benchmark controller \(C_\text{L}\) are shown in Fig. \ref{fig: Sensitivity all}. It is observed that \(C_{\mathcal{R}_1}\) yields improved low-frequency sensitivity attenuation while maintaining the peak sensitivity below the imposed robustness bound of 6 dB.  

\begin{table}[]\centering
\caption{Parameters of the \(C_{\mathcal{R}_1}\) and \(C_{\mathcal{R}_2}\) controllers. Details of the structure and parameter selection, particularly for the CgLp element, are provided in \citep[Definition 1]{hosseini2025AddOnFilterDesign}. Frequencies are given in rad/s.}
\label{tab: controller parameters}
\resizebox{0.75\columnwidth}{!}{%
\begin{tabular}{cccccc}
\hline
          & {$\omega_r$} &{$\omega_l$} & {$\omega_f$} & {$A_\rho$}& {$\omega_i$}    \\ \hline
{$C_{\mathcal{R}_1}$}       & {$0.0703$}        &{$0.1138$}& {$0.1963$}                &{$0$}& $0.0439 $ \\ 
{$C_{\mathcal{R}_2}$}       & {$0.0485$}        &{$0.0785$}& {$0.2301$}                &{$0$}& {$0.0746$} \\ \hline
\end{tabular}
}
\end{table}

\begin{figure} 
\centering
\includegraphics[width=0.95\columnwidth]{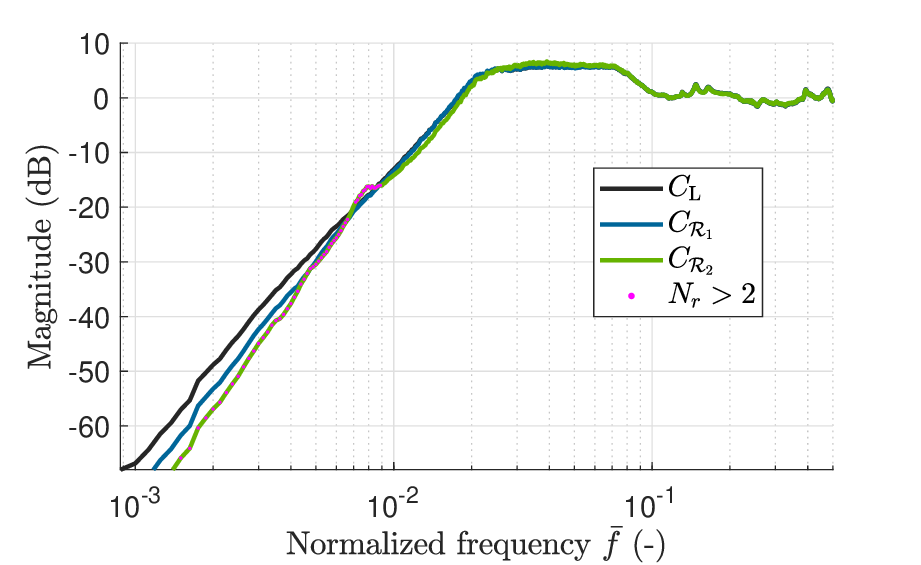}
\caption{The (pseudo)-sensitivity magnitude of the linear (\(C_\text{L}\)) controller, reset controller \(C_{\mathcal{R}_1}\) and reset controller \(C_{\mathcal{R}_2}\). }
\label{fig: Sensitivity all}
\end{figure}

To illustrate the effect of further increasing the integral corner frequency beyond that achieved with \(C_{\mathcal{R}_1}\), a second reset controller, denoted by \(C_{\mathcal{R}_2}\), is considered. In this case, the integral corner frequency is increased to \(
\omega_i = 2\pi \times 11.875 \times 10^{-3}\ \text{rad/s},
\) which is nearly three times that of the linear controller. Consequently, a larger amount of nonlinear phase compensation is required from the CgLp element. The pseudo-sensitivity corresponding to \(C_{\mathcal{R}_2}\) is also shown in Fig. \ref{fig: Sensitivity all}. Although this controller achieves the highest \(\omega_i\) among all considered designs, its sensitivity magnitude exceeds that of the linear controller over part of the frequency range. This indicates that merely increasing \(\omega_i\) by introducing stronger nonlinear phase compensation does not necessarily yield improved closed-loop performance, as the nonlinear effects induced by the reset element may compromise the anticipated performance gains.

Furthermore, the regions marked with red dots indicate frequency intervals in which more than two reset instants occur within a single period. In these regions, the frequency-domain closed-loop prediction is no longer reliable, as it is derived under the assumption of exactly two zero crossings per period. 

\begin{table} \centering
\caption{Parameters of the fourth-order shaping filter \(C_{s_4}\) designed for the reset element, with \(\omega_r = 2\pi \times 7.72 \times 10^{-3}\,\mathrm{rad/s}\). Here, \(a_0 = 1\) and \(b_0 = 1\).}
\label{tab: Cs4 parameters}
\resizebox{0.7\columnwidth}{!}{%
\begin{tabular}{cccccccc}

          & {$a_1$} &{$a_2$} & {$a_3$} & {$a_4$} \\ \hline
      & {$0.7178$}        &{$0.4034$}& {$0.1418$}                &{$0.0269$} \\ \hline
      \\
      &{$b_1$}& {$b_2$} &{$b_3$}           &{$b_4$}    \\ \hline
      &{$1.7178$}&{$0.9976$} &{$0.3329$}         & {$0.0537$} \\
\hline
\end{tabular}
}
\end{table}

\subsubsection{Reset control with shaping filter}
Here, we reconsider the reset controller \(C_{\mathcal{R}_2}\) introduced in the previous subsection. 
All controller parameters are kept unchanged, but \(C_s=1\) is replaced by a fourth-order shaping filter in the reset-triggering path.  The resulting shaped reset controller is denoted by \(C_{\mathcal{R}_2}^{\mathrm{SF}}\). As shown previously, \(C_{\mathcal{R}_2}\) achieves a higher integrator corner frequency, \(\omega_i = 2\pi \times 11.875 \times 10^{-3}\ \mathrm{rad/s}\), but at the cost of large HOSIDF magnitudes. To reduce this nonlinear contribution, a fourth-order shaping filter is introduced. This order is selected because first-, second-, and third-order shaping filters were found insufficient to reduce the nonlinearity to the desired level for this controller, whereas higher-order filters provided no substantial additional benefit while unnecessarily increasing the controller order and implementation complexity. The shaping filter is given by
\begin{equation}
    \label{eq: Cs4}
    C_{s_4}(s)=
    \frac{\sum_{k=0}^{4} a_k\,\omega_r^{-k} s^{k}}
         {\sum_{k=0}^{4} b_k\,\omega_r^{-k} s^{k}},
\end{equation}
where \(\omega_r = 2\pi \times 7.72 \times 10^{-3}\ \text{rad/s}.\)
The filter is designed to reduce the degree of nonlinearity in the frequency regions characterized by more than two zero-crossings, thereby improving the behavior of both the error signal \(e(t)\) and the reset input \(e_r(t)\). Following Theorem~\ref{theorem: shaping filter}, the corresponding parameter values are listed in Table~\ref{tab: Cs4 parameters}. The shaping filter itself, together with the base linear transfer function of the reset element, i.e., \(R(j\omega)\) from \eqref{eq RCS bls} with \(D_r=0\), is shown in Fig.~\ref{fig: Cs4}. The filter is designed such that, in the frequency region where the nonlinearity level is excessively high and multiple zero-crossings occur for \(C_{\mathcal{R}_2}\), its phase matches that of the reset element. Consequently, \(e_r(t)\) and \(u_r(t)\) from Fig.~\ref{fig: reset block diagram} become approximately in phase. This implies that when \(e_r(t)\) crosses zero, the reset output \(u_r(t)\) is also close to zero. Therefore, the reset jump is either eliminated or becomes negligibly small at those frequencies, which in turn reduces the nonlinearity.

\begin{figure}
\centering
\includegraphics[width=0.95\columnwidth]{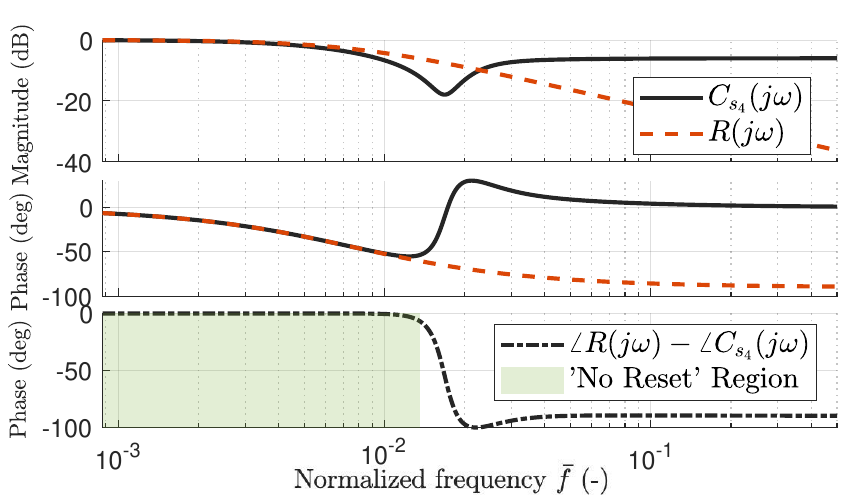}
\caption{The designed fourth-order shaping filter for the reset controller \(C_{\mathcal{R}_2}^{\mathrm{SF}}\).}
\label{fig: Cs4}
\end{figure}

\begin{figure}
\centering
\includegraphics[width=0.95\columnwidth]{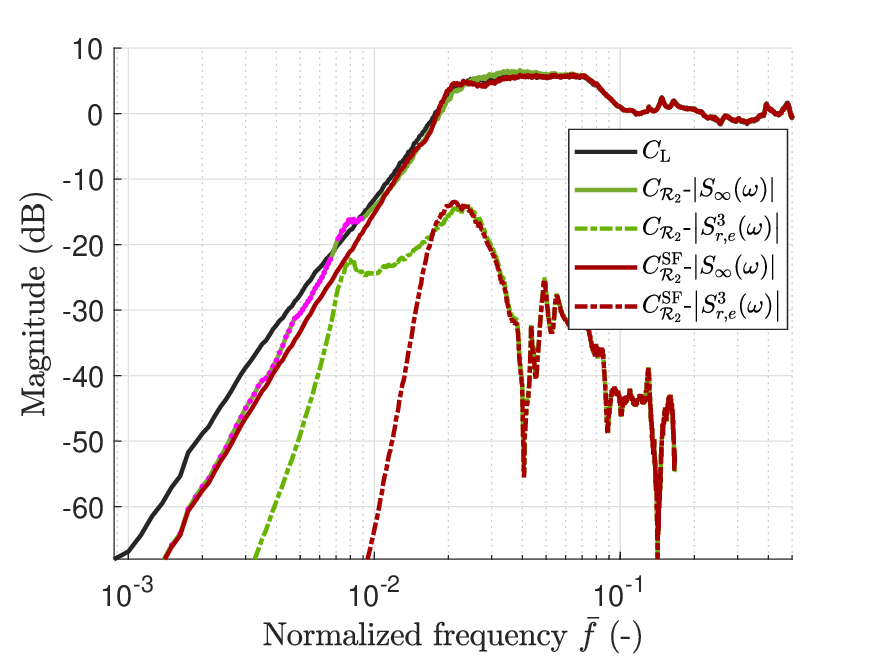}
\caption{The (pseudo-)sensitivity magnitudes of the linear controller \(C_\text{L}\), the reset controller \(C_{\mathcal{R}_2}\) and reset controller $C_{\mathcal{R}_2}^{\mathrm{SF}}$ including the shaping filter, together with their corresponding third-order sensitivities.}
\label{fig: Sensitivity SF}
\end{figure}

Fig. \ref{fig: Sensitivity SF} shows the pseudo-sensitivity of the controllers \(C_{\mathcal{R}_2}\) and $C_{\mathcal{R}_2}^{\mathrm{SF}}$. It can be observed that, when the shaping filter is included, no frequency region exhibits more than two zero-crossings. As a result, the sensitivity magnitude remains below that of the unshaped case across the entire frequency range. To further illustrate how the shaping filter mitigates the nonlinear behavior, the third-order sensitivity of the error signal, \(\left|S_{r,e}^{3}(\omega)\right|\), is also shown for designed controllers. The results indicate that the HOSIDF magnitude is reduced in the frequency regions where the phase of the shaping filter is properly aligned with that of the base linear dynamics of the reset element (the ``No-Reset region'' indicated in Fig. \ref{fig: Cs4}).

\subsection{Experimental Validation}
In this section, the controllers \(C_{\mathrm{L}}\), \(C_{\mathcal{R}_1}\), \(C_{\mathcal{R}_2}\) and $C_{\mathcal{R}_2}^{\mathrm{SF}}$ are experimentally implemented on the considered wire bonder. Here, \(C_{\mathrm{L}}\) denotes the benchmark linear controller, \(C_{\mathcal{R}_1}\) represents the best-performing reset controller without a shaping filter, and \(C_{\mathcal{R}_2}^{\mathrm{SF}}\) corresponds to the reset controller augmented with the fourth-order shaping filter in \eqref{eq: Cs4}.

Based on the frequency-domain analysis reported in Fig.~\ref{fig: Sensitivity SF}, three excitation frequencies are selected to experimentally characterize the closed-loop behavior of the wire bonder under a sinusoidal reference input,
\(
r(t)=\hat{r}\sin\left(2\pi f_{\mathrm{in}}t\right).
\)
The first frequency, \(f_{\mathrm{in}}=2.5\times10^{-3}\,\)Hz, is chosen in the low-frequency range to verify the improvement in tracking performance without introducing additional nonlinear effects. The second frequency, \(f_{\mathrm{in}}=7.5\times10^{-3}\,\)Hz, lies in the mid-frequency range, where the nonlinear behavior of the conventional reset element was found to be most pronounced. The third frequency is selected close to the closed-loop crossover frequency \(f_{\mathrm{in}}=2.5\times10^{-2}\,\)Hz in order to verify that the predicted robustness margin is preserved for all controllers. All measured results are again normalized; therefore, the amplitude \(\hat{r}\) is omitted from the presentation. Nevertheless, to provide physical context, the resulting tracking errors may be interpreted relative to the sub-micrometer positioning-accuracy requirements of the machine.

In Fig. \ref{fig: main measurments}, the measured error signals, \(e(t)\), and the measured reset output signals, \(u_r(t)\), are presented for three selected input frequencies. The error responses in Figs. \ref{fig:e20}, \ref{fig:e60}, and \ref{fig:e200} show that the reset controller \(C_{\mathcal{R}_2}^{\mathrm{SF}}\) achieves a smaller error in the low- and mid-frequency regions, while exhibiting a comparable error magnitude to the other controllers near the sensitivity peak. Furthermore, the reset output signals in Figs. \ref{fig:ur20}, and \ref{fig:ur60} show that, for \(C_{\mathcal{R}_2}^{\mathrm{SF}}\), zero crossings of \(e_r(t)\) either disappear or occur when the reset state is close to zero. Consequently, the reset jump becomes negligible, and the reset element behaves close to its base linear dynamics, whereas reset action is still activated at \(2.5\times10^{-2}\,\mathrm{Hz}\). This experimentally validates the theoretical results presented in Section \ref{Sec: Smart SF}, showing that the proposed shaping-filter design can render the reset element nearly inactive in frequency regions where reset-induced nonlinearity is not required, while preserving its effectiveness in the frequency region where reset action is desired, as demonstrated here at \(2.5\times10^{-2}\,\mathrm{Hz}\).

\begin{figure*}
    \centering

    \begin{subfigure} {0.48\linewidth}
        \centering
        \includegraphics[width=\linewidth]{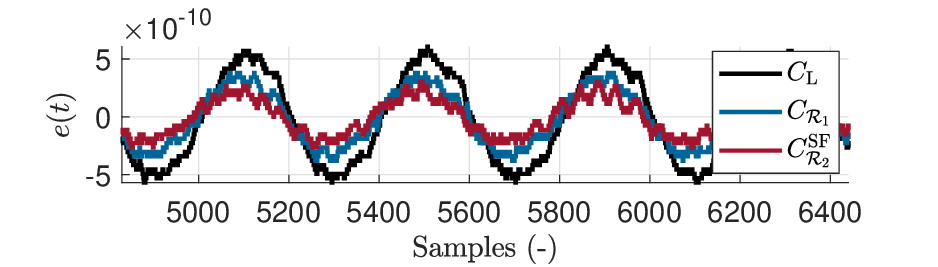}
        \caption{Error at $2.5\times10^{-3}$ Hz}
        \label{fig:e20}
    \end{subfigure}
    \hfill
    \begin{subfigure} {0.48\linewidth}
        \centering
        \includegraphics[width=\linewidth]{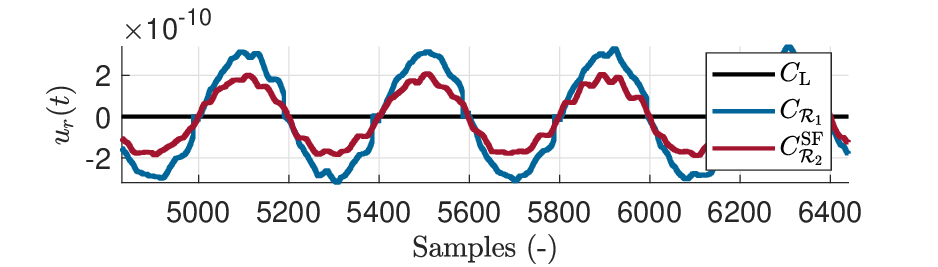}
        \caption{Reset output at $2.5\times10^{-3}$ Hz}
        \label{fig:ur20}
    \end{subfigure}

    \vspace{0.5em}

    \begin{subfigure} {0.48\linewidth}
        \centering
        \includegraphics[width=\linewidth]{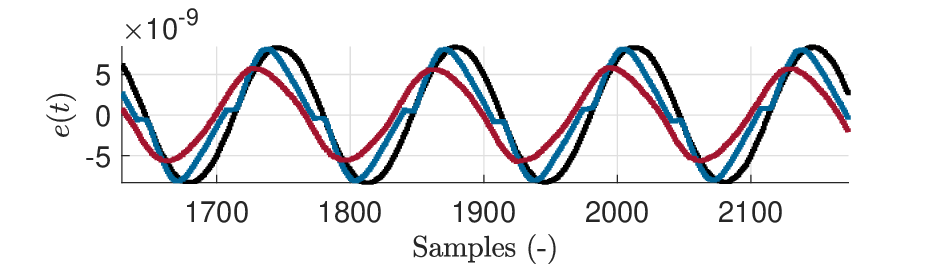}
        \caption{Error at $7.5\times10^{-3}$ Hz}
        \label{fig:e60}
    \end{subfigure}
    \hfill
    \begin{subfigure} {0.48\linewidth}
        \centering
        \includegraphics[width=\linewidth]{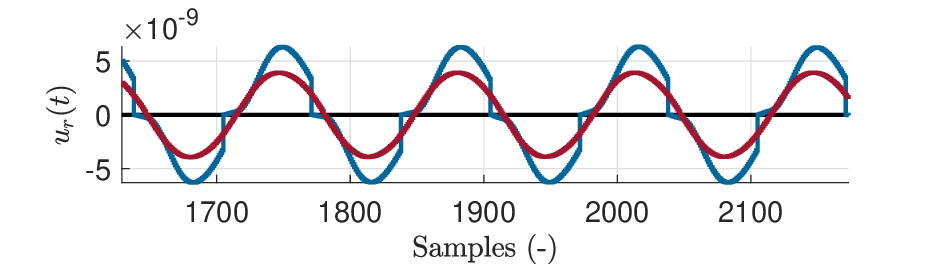}
        \caption{Reset output at $7.5\times10^{-3}$ Hz}
        \label{fig:ur60}
    \end{subfigure}

    \vspace{0.5em}

    \begin{subfigure}{0.48\linewidth}
        \centering
        \includegraphics[width=\linewidth]{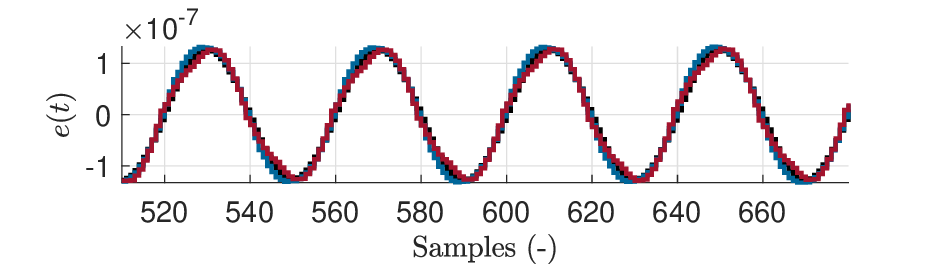}
        \caption{Error at $2.5\times10^{-2}$ Hz}
        \label{fig:e200}
    \end{subfigure}
    \hfill
    \begin{subfigure}{0.48\linewidth}
        \centering
        \includegraphics[width=\linewidth]{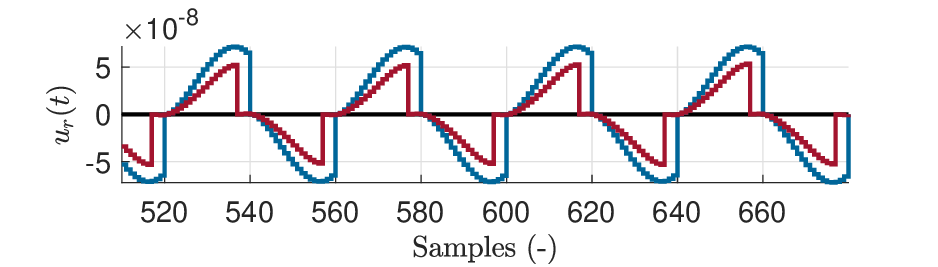}
        \caption{Reset output at $2.5\times10^{-2}$ Hz}
        \label{fig:ur200}
    \end{subfigure}

   \caption{Normalized error ($e(t)$) and reset output ($u_r(t)$) for \(C_{\mathrm{L}}\), \(C_{\mathcal{R}_1}\), and $C_{\mathcal{R}_2}^{\mathrm{SF}}$, shown at frequencies of (a,b) $2.5\times10^{-3}$ Hz, (c,d) $7.5\times10^{-3}$ Hz, and (e,f) $2.5\times10^{-2}$ Hz.}
    \label{fig: main measurments}
\end{figure*}

To further assess the effectiveness of the proposed shaping filters in reducing nonlinear effects, the power spectral density (PSD) and cumulative power spectral density (CPSD) of the error signals for \(f_{\mathrm{in}}=7.5\times10^{-3}\,\mathrm{Hz}\) are shown in Fig. \ref{fig: CPSD and PSD}. As observed from Fig. \ref{fig: CPSD60}, \(C_{\mathcal{R}_2}^{\mathrm{SF}}\) provides a clear improvement in error reduction compared with the other controllers. In Fig. \ref{fig: PSD60}, the PSD of the error signals is shown after the first spectral component, in order to emphasize the contributions associated with the higher-order harmonics (i.e., \(3f_{\mathrm{in}}\), \(5f_{\mathrm{in}}\), \(\ldots\)). Although \(C_{\mathcal{R}_2}^{\mathrm{SF}}\) employs a stronger reset element to compensate for a higher integrator corner frequency, it exhibits significantly lower higher-order harmonic content than \(C_{\mathcal{R}_1}\). This reduction is attributed to the proposed shaping filter, which effectively suppresses reset-induced nonlinearities while retaining the desired reset action in the targeted frequency range.

To explicitly demonstrate the effectiveness of the shaping filter, Fig.~\ref{fig: erNoCs_comparison} compares the measured reset-triggering signal \(e_r(t)\) of controllers \(C_{\mathcal{R}_2}\) and $C_{\mathcal{R}_2}^{\mathrm{SF}}$. The results indicate that, across all analyzed input frequencies, the shaping filter improves the behavior of the reset-triggering signal. This is important because higher-order harmonics, noise, or other unwanted frequency components in \(e_r(t)\) can introduce multiple zero crossings and consequently degrade the closed-loop performance. In particular, for \(f_{\mathrm{in}}=2.5\times10^{-3}\,\mathrm{Hz}\) and \(f_{\mathrm{in}}=7.5\times10^{-3}\,\mathrm{Hz}\), the shaping filter effectively suppresses excessive zero crossings, thereby improving the reliability of the reset control system. This improvement is especially evident for \(f_{\mathrm{in}}=7.5\times10^{-3}\,\mathrm{Hz}\), where the controller \(C_{\mathcal{R}_2}\) exhibits six zero crossings within one period.

To complement the time-domain and frequency-domain plots, 
Table~\ref{tab:experimental_metrics} reports quantitative performance 
metrics for all considered controllers. For each excitation frequency, 
the RMS value of the normalized tracking error is computed as
\begin{equation}
    {e}_{\mathrm{RMS}}
    =
    \sqrt{\frac{1}{N}\sum_{k=1}^{N}{e}^2[k]},
\end{equation}
To quantify the cumulative higher-frequency content in the reset-triggering
signal, the CPSD after the fundamental excitation frequency is evaluated. For
each controller \(C\), this quantity is computed as
\begin{equation}
    \mathrm{CPSD}_{e_r,C}^{> f_{\mathrm{in}}}
    =
    \int_{ f_{\mathrm{in}}^{+}}^{ f_{\mathrm{Nyq}}}
    \mathrm{PSD}_{e_r,C}(\nu)\,d\nu ,
\end{equation}
where \( f_{\mathrm{in}}^{+}\) denotes a frequency slightly above the fundamental excitation frequency, chosen to exclude the dominant fundamental peak, and $f_{\mathrm{Nyq}}$ is the normalized Nyquist frequency. The reported value is normalized with respect to the corresponding value obtained for
\(C_{\mathcal{R}_1}\), i.e.,
\begin{equation}
   \frac{
    \mathrm{CPSD}_{e_r,C}^{> f_{\mathrm{in}}}
    }{
    \mathrm{CPSD}_{e_r,C_{\mathcal{R}_1}}^{>f_{\mathrm{in}}}
    } .
\end{equation}
Therefore, the value for \(C_{\mathcal{R}_1}\) is equal to one, while values below one indicate a reduction of the cumulative higher-frequency content in
\(e_r(t)\) compared with \(C_{\mathcal{R}_1}\).

The results in Table~\ref{tab:experimental_metrics} show that increasing the reset-based phase compensation without a shaping filter can reduce the tracking error at low and intermediate excitation frequencies, but it also increases the cumulative higher-frequency content in the reset-triggering signal, resulting in unreliability of the predicted performance. This is observed for \(C_{\mathcal{R}_2}\), whose CPSD ratio is larger than one for all tested frequencies, and particularly large at \({f}_{\mathrm{in}}=7.5\times 10^{-3}\). In contrast, the proposed shaped controller \(C_{\mathcal{R}_2}^{\mathrm{SF}}\) consistently yields CPSD ratios below one, indicating that the shaping filter suppresses the post-fundamental content of \(e_r(t)\) relative to \(C_{\mathcal{R}_1}\). At the same time, \(C_{\mathcal{R}_2}^{\mathrm{SF}}\) achieves the lowest normalized RMS tracking error among the tested controllers for all three excitation frequencies. This confirms that the shaping filter improves the reliability of the reset action while preserving, and in these experiments improving, the tracking performance.

Overall, the experimental results confirm that the proposed shaping-filter-based reset controller \(C_{\mathcal{R}_2}^{\mathrm{SF}}\) improves the tracking performance while effectively suppressing undesired reset-induced nonlinearities. The reset action is selectively activated in the intended frequency region, thereby preserving the desired phase-compensation benefit without introducing excessive higher-order harmonic content. These results validate the practical applicability of the proposed design presented in Section~\ref{Sec: Smart SF} on the considered wire-bonder system.

\begin{table}[t]
\centering
\caption{Experimental performance comparison of the considered controllers at different excitation frequencies. The RMS value is computed from the normalized tracking error \({e}(t)\). The CPSD ratio is computed from the reset-triggering signal \(e_r(t)\) after the fundamental excitation frequency and is normalized by the corresponding value of \(C_{\mathcal{R}_1}\).}
\label{tab:experimental_metrics}
\begin{tabular}{cccc}
\toprule
$f_\mathrm{in}$ & Controller 
& \({e}_{\mathrm{RMS}}\) 
& \( \frac{
    \mathrm{CPSD}_{e_r,C}^{> f_{\mathrm{in}}}
    }{
    \mathrm{CPSD}_{e_r,C_{\mathcal{R}_1}}^{>f_{\mathrm{in}}}
    }\) \\
\midrule

\multirow{4}{*}{\( 2.5\times 10^{-3}\)}
& \(C_{\mathrm{L}}\) 
& \(3.735\times 10^{-10}\) & -- \\
& \(C_{\mathcal{R}_1}\) 
& \(2.367\times 10^{-10}\) & \(1\) \\
& \(C_{\mathcal{R}_2}\) 
& \(1.531\times 10^{-10}\) & \(1.463\) \\
& \(C_{\mathcal{R}_2}^{\mathrm{SF}}\) 
& \(1.520\times 10^{-10}\) & \(0.1856\) \\
\midrule

\multirow{4}{*}{\( 7.5\times 10^{-3}\)}
& \(C_{\mathrm{L}}\) 
& \(5.932\times 10^{-9}\) & -- \\
& \(C_{\mathcal{R}_1}\) 
& \(5.025\times 10^{-9}\) & \(1\) \\
& \(C_{\mathcal{R}_2}\) 
& \(4.285\times 10^{-9}\) & \(4.412\) \\
& \(C_{\mathcal{R}_2}^{\mathrm{SF}}\) 
& \(3.908\times 10^{-9}\) & \(0.0641\) \\
\midrule

\multirow{4}{*}{\( 2.5\times 10^{-2}\)}
& \(C_{\mathrm{L}}\) 
& \(9.018\times 10^{-8}\) & -- \\
& \(C_{\mathcal{R}_1}\) 
& \(9.678\times 10^{-8}\) & \(1\) \\
& \(C_{\mathcal{R}_2}\) 
& \(1.000\times 10^{-7}\) & \(1.275\) \\
& \(C_{\mathcal{R}_2}^{\mathrm{SF}}\) 
& \(8.921\times 10^{-8}\) & \(0.1651\) \\
\bottomrule
\end{tabular}
\end{table}

\subsubsection{Implementation Aspects of Linear and Reset Controllers}

The linear and reset controllers considered in this study are implemented digitally. All LTI components are discretized using the Tustin approximation, which is known to provide favorable phase preservation of the corresponding continuous-time dynamics compared with alternative discretization methods, particularly over a broad frequency range up to frequencies close to the Nyquist limit~(\cite{aastrom2013computerTustin}). Accordingly, the frequency-response results reported for the linear components are also evaluated in discrete time, i.e., based on the discretized frequency-response functions. The reset element is also discretized using the Tustin method. The resulting discrete-time realization of the GFORE and its implementation details are discussed in~\cite[Section~V.B]{hosseini2025AddOnFilterDesign}.

Since the frequency axis is normalized throughout the analysis, it is useful to note that the sampling frequency used in the experimental implementation is on the order of \(10\,\mathrm{kHz}\). This information facilitates the physical interpretation of the reported frequency-domain results. Furthermore, as discussed earlier, the measured tracking errors should be interpreted in relation to the sub-micrometer positioning accuracy requirements of the considered industrial motion system.

Also, closed-loop stability of the implemented reset controllers was verified using the same stability assessment as in \cite{dastjerdi2023frequency}, and all experimental tests were conducted within the stable operating regime. A full stability synthesis for shaping filter-based reset controllers is beyond the scope of this paper.

\begin{figure}
\centering
\subfloat[]{\includegraphics[width=0.495\columnwidth]{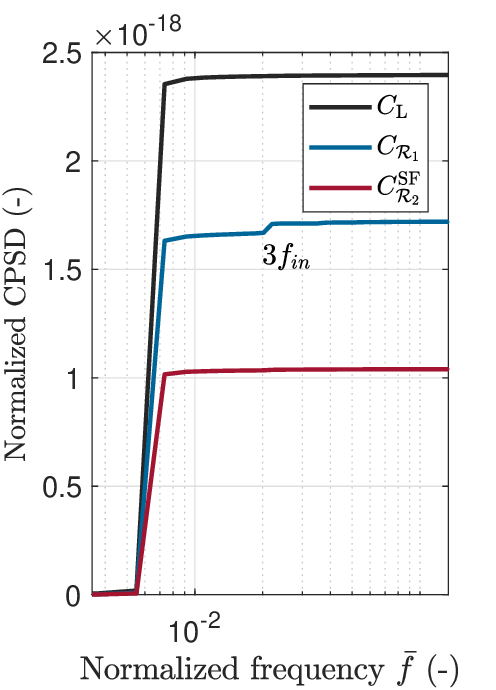}%
\label{fig: CPSD60}}
\hfil
\subfloat[]{\includegraphics[width=0.47\columnwidth]{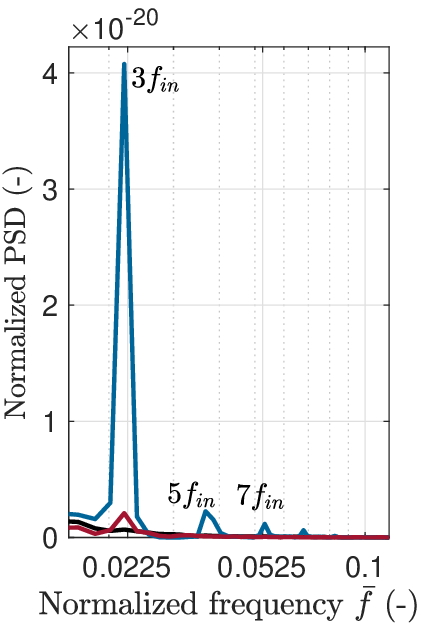}%
\label{fig: PSD60}}
\caption{(a) Cumulative PSD of the errors for an input at $7.5\times10^{-3}$ Hz. (b) PSD of the errors after the first jump in the cumulative PSD at $7.5\times10^{-3}$ Hz, shown to compare the presence of higher-order harmonics.}
\label{fig: CPSD and PSD}
\end{figure}

\begin{figure}
    \centering

    \begin{subfigure}{\linewidth}
        \centering
        \includegraphics[width=\linewidth]{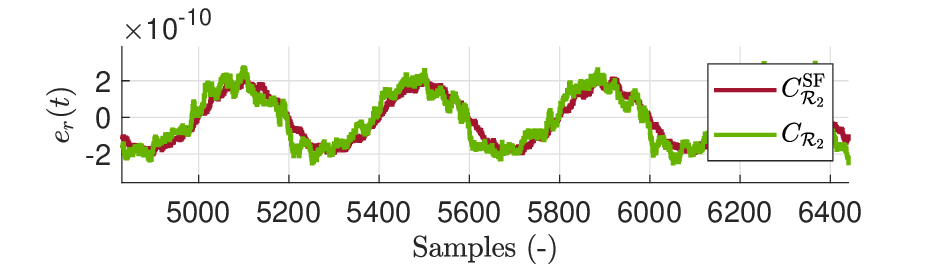}
        \caption{$2.5\times10^{-3}\,$Hz}
        \label{fig: erNoCs_20}
    \end{subfigure}

    \begin{subfigure}{\linewidth}
        \centering
        \includegraphics[width=\linewidth]{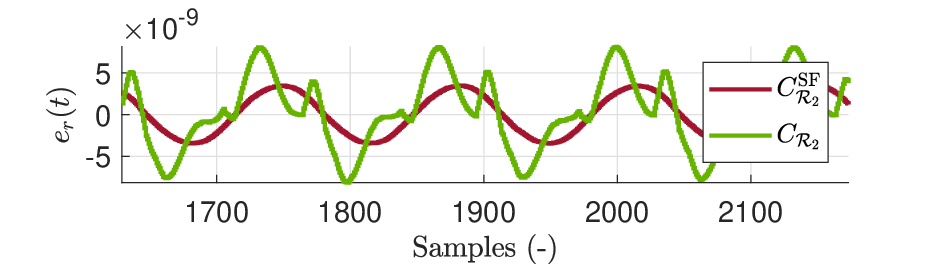}
        \caption{$7.5\times10^{-3}\,$Hz}
        \label{fig: erNoCs_60}
    \end{subfigure}

    \begin{subfigure}{\linewidth}
        \centering
        \includegraphics[width=\linewidth]{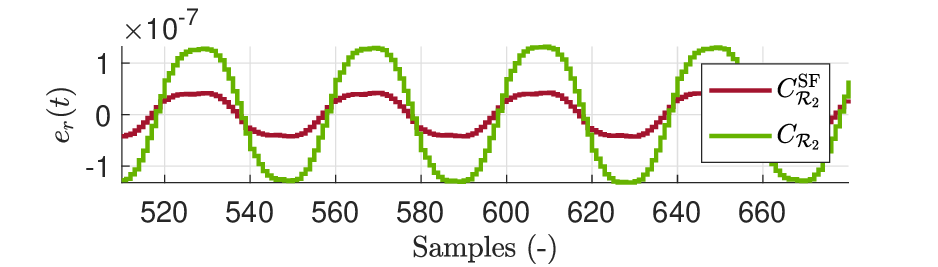}
        \caption{$2.5\times10^{-2}\,$Hz}
        \label{fig: erNoCs_200}
    \end{subfigure}

    \caption{Measured reset-triggering signal \(e_r(t)\) for controllers \(C_{\mathcal{R}_2}\) and $C_{\mathcal{R}_2}^{\mathrm{SF}}$, comparing the cases with and without the proposed shaping filter $C_{s_4}$.}
    \label{fig: erNoCs_comparison}
\end{figure}

%% file: Sections/06Conclusion.tex
\section{Conclusion}\label{sec: conclusion}
This paper proposes a systematic shaping filter design framework for the generalized first-order reset element to reduce undesired reset-induced nonlinearities. By deriving explicit conditions on the shaping-filter coefficients, the proposed method increases the low-frequency attenuation slope of the HOSIDFs as the filter order increases, thereby providing stronger suppression of higher-order harmonics in frequency regions where reset action is not desired.

The effectiveness of the method was demonstrated through frequency-domain analysis, a superposition-law assessment, and experimental validation on an industrial wire-bonder motion stage. The results show that the proposed shaping-filter-based reset controller improves tracking performance, suppresses excessive reset actions, and enhances the predictability and practical reliability of reset control systems. In the wire-bonder experiments, the proposed shaped reset controller reduced the RMS tracking error by up to approximately \(59\%\) compared with the linear benchmark and reduced the post-fundamental reset-triggering CPSD by more than \(80\%\) relative to the best unshaped reset controller. Future work will focus on automated tuning of the shaping filter parameters.

%% file: Sections/08Appendix_B.tex
\section{Proof of Theorem \ref{theorem: shaping filter}}\label{App: pf of theorem}
In this proof, we show that for the given shaping filter, with the corresponding coefficients $a_k$ and $b_k$, the resulting value is $\upsilon_0 = 2p + 1$, where, according to Lemma~\ref{lem: slop Upsilon}, we have
\begin{equation}
    \label{eq: var upsilon 0 pf1}
    \upsilon_0=\lim_{\omega \to 0} 
\frac{d\,\log_{10}|\Upsilon(\omega)|}{d\,\log_{10}\omega}.
\end{equation}
Therefore, from \eqref{eq: HOSIDFs all functions} and $A_r=-\omega_r$, we have
\begin{equation}
    \label{eq: mag Upsilon}
\Upsilon(\omega)=\omega\cos\left(\varphi(\omega)\right) +\omega_r \sin\left(\varphi(\omega)\right).
\end{equation}
Since $\varphi(\omega)=\arg\left(C_s(j\omega)\right)$, we have $\varphi(\omega)=\operatorname{atan2}\left(\zeta(\omega)\right)$, with
\begin{equation}
\label{eq: zeta}
\zeta(\omega)=\frac{\ImPart{C_s(j\omega)}}{\RePart{C_s(j\omega)}},
\end{equation}
where $\RePart{.}$ and $\ImPart{.}$ denote the real and imaginary parts of a complex value, respectively. Having $C_s(j\omega)$ from \eqref{eq: Cs}, for simplicity in this proof we consider $\alpha_k=a_k\omega_r^{-k}$ and $\beta_k=b_k\omega_r^{-k}$, and having $s=j\omega$, it gives
\begin{equation}
    \label{eq: Cs proof}
    C_s(j\omega) = 
    \frac{\sum_{k=0}^{p} \alpha_k (j\omega)^{k}}
         {\sum_{k=0}^{p} \beta_k (j\omega)^{k}}.
\end{equation}
In the following, we calculate the real and imaginary parts of $C_s(j\omega)$. Having the term $(j)^k$ as
\begin{equation}
    (j)^k = \cos\!\left(\frac{\pi}{2}k\right) + j \sin\!\left(\frac{\pi}{2}k\right),
\end{equation}
the $C_s(j\omega)$ in \eqref{eq: Cs proof} can be written as
\begin{equation}
    \label{eq: Cs proof v2}
    C_s(j\omega) = 
    \frac{\sum_{k=0}^{p} \alpha_k \omega^k\cos\!\left(\frac{\pi}{2}k\right)+j\sum_{k=0}^{p} \alpha_k \omega^k\sin\!\left(\frac{\pi}{2}k\right)}
         {\sum_{k=0}^{p} \beta_k \omega^k\cos\!\left(\frac{\pi}{2}k\right)+j\sum_{k=0}^{p} \beta_k \omega^k\sin\!\left(\frac{\pi}{2}k\right)}.
\end{equation}
Thus, we can write
\begin{flalign}
\RePart{&C_s(j\omega)} =
\notag\\
&\frac{
\sum_{l=0}^p \sum_{k=0}^p
\alpha_k \beta_l \omega^{k+l}
\cos\left(\frac{(k-l)\pi}{2}\right)
}
{
\left( \sum_{k=0}^p \beta_k \omega^k
\cos\left(\frac{\pi}{2}k\right) \right)^2
+
\left( \sum_{k=0}^p \beta_k \omega^k
\sin\left(\frac{\pi}{2}k\right) \right)^2
}, &&
\label{eq: real Cs}
\end{flalign}
and
\begin{flalign}
\ImPart{&C_s(j\omega)} =
\notag\\
&\frac{\sum_{l=0}^p \sum_{k=0}^p \alpha_k  \beta_l \omega^{k+l}\sin\left(\frac{(k-l)\pi}{2}\right)}
{\left( \sum_{k=0}^p \beta_k \omega^k \cos\left(\frac{\pi}{2}k\right) \right)^2
+
\left( \sum_{k=0}^p \beta_k \omega^k \sin\left(\frac{\pi}{2}k\right) \right)^2}. &&
\label{eq: imag Cs}
\end{flalign}
Now having the real and imaginary parts of $C_s(j\omega)$, we rewrite \eqref{eq: mag Upsilon} as
\begin{equation}
    \label{eq: mag Upsilon v2}
\Upsilon(\omega)=\omega\cos\left(\operatorname{atan2}\left(\zeta(\omega)\right)\right) +\omega_r \sin\left(\operatorname{atan2}\left(\zeta(\omega)\right)\right).
\end{equation}
Using $\cos(\operatorname{atan2} (x))=\tfrac{1}{\sqrt{1+x^2}}$ and $\sin\left(\operatorname{atan2} (x)\right)=\tfrac{x}{\sqrt{1+x^2}}$, \eqref{eq: mag Upsilon v2} is rewritten as
\begin{equation}
    \label{eq: mag Upsilon v3}
\Upsilon(\omega)=\frac{\omega+\omega_r\zeta(\omega)}{\sqrt{1+\zeta(\omega)^2}},
\end{equation}
and replacing $\zeta(\omega)=\frac{\ImPart{C_s(j\omega)}}{\RePart{C_s(j\omega)}}$ we get
\begin{equation}
    \label{eq: mag Upsilon v4}
\Upsilon(\omega)=\frac{\omega\RePart{C_s(j\omega)}+\omega_r\ImPart{C_s(j\omega)}}{\RePart{C_s(j\omega)}\sqrt{1+\frac{\ImPart{C_s(j\omega)}}{\RePart{C_s(j\omega)}}}}.
\end{equation}
Regarding \eqref{eq: var upsilon 0 pf1}, at the end, we are interested in $\omega\rightarrow0$. Thus having $\lim_{\omega \to 0}\RePart{C_s(j\omega)}=\alpha_0$ and $\lim_{\omega \to 0}\ImPart{C_s(j\omega)}=0$, we can write
\begin{equation}
    \label{eq: mag Upsilon v5}
\lim_{\omega \to 0}\Upsilon(\omega)=\frac{1}{\alpha_0}\lim_{\omega \to 0}{\omega\RePart{C_s(j\omega)}+\omega_r\ImPart{C_s(j\omega)}}.
\end{equation}
Since the denominator of both real and imaginary parts from \eqref{eq: real Cs} and \eqref{eq: imag Cs} are equal to $\beta_0$ as $\omega\rightarrow0$, we have
\begin{equation}
    \label{eq: mag Upsilon v6}
\lim_{\omega \to 0}\Upsilon(\omega)=\frac{1}{\alpha_0\beta_0}\lim_{\omega \to 0}\Pi(\omega),
\end{equation}
where
\begin{align}
    \label{eq: Pi}
    \Pi(\omega)=&\sum_{l=0}^p \sum_{k=0}^p \alpha_k  \beta_l \omega^{k+l+1}\cos\left(\frac{(k-l)\pi}{2}\right)\nonumber\\
    &+\omega_r\sum_{l=0}^p \sum_{k=0}^p \alpha_k  \beta_l \omega^{k+l}\sin\left(\frac{(k-l)\pi}{2}\right).
\end{align}
It can be seen that because of the presence of sin and cos functions, $\Pi(\omega)$ contains only terms with $\omega^{2m+1}$ with $m=[1,p]\in\mathbb{N}$. Since $\Pi(\omega)$ directly effects the $\Upsilon(\omega)$, we aim to keep $\omega^{2p+1}$ and set all other coefficients for any terms $\omega^1,\omega^3,...,\omega^{2p-1}$ to zero to be able to get the maximum possible value for $\upsilon_0$ at \eqref{eq: var upsilon 0 pf1}. Thus, \eqref{eq: Pi} can be written as
\begin{equation}
    \label{eq: Pi 2 proof}
\Pi(\omega)=\alpha_p\beta_p\omega^{2p+1}+\sum_{m=1}^p\omega^{2m-1}Q_m(\omega)
\end{equation}
where $\forall\,\,m\in[1,p]$,
\begin{align}
\label{eq: Eq m}
    Q_m(\omega)=&\mathop{\sum_{l=0}^p \sum_{k=0}^p}_{k+l=2m-2} \alpha_k  \beta_l \cos\left(\frac{(k-l)\pi}{2}\right) \nonumber\\
&+\omega_r\mathop{\sum_{l=0}^p \sum_{k=0}^p}_{k+l=2m-1} \alpha_k  \beta_l \sin\left(\frac{(k-l)\pi}{2}\right).
\end{align}
Having $\alpha_k=a_k\omega_r^{-k}$ and $\beta_l=b_l\omega_r^{-l}$, it gives
\begin{align}
\label{eq: Eq m v3}
    Q_m(\omega)=&\omega_r^{-(2m-2)}\Bigg(\mathop{\sum_{l=0}^p \sum_{k=0}^p}_{k+l=2m-2} a_k  b_l \cos\left(\frac{(k-l)\pi}{2}\right) \nonumber\\
&+\mathop{\sum_{l=0}^p \sum_{k=0}^p}_{k+l=2m-1} a_k  b_l \sin\left(\frac{(k-l)\pi}{2}\right)\Bigg).
\end{align}
Thus, $Q_m(\omega)=0$ if
\begin{align}
    \label{eq: Eq m v4}
   \mathop{\sum_{l=0}^p \sum_{k=0}^p}_{k+l=2m-2} &a_k  b_l \cos\left(\frac{(k-l)\pi}{2}\right)\notag\\
&+\mathop{\sum_{l=0}^p \sum_{k=0}^p}_{k+l=2m-1} a_k  b_l \sin\left(\frac{(k-l)\pi}{2}\right)=0,
\end{align}
or equivalently (straightforward to show by considering $k-l=2m-2-2l$ and $k-l=2m-1-2l$),
\begin{equation}
    \label{eq: Eq m v5}
   \mathop{\sum_{l=0}^p \sum_{k=0}^p}_{k+l=2m-2} a_k  b_l (-1)^l
+\mathop{\sum_{l=0}^p \sum_{k=0}^p}_{k+l=2m-1} a_k  b_l (-1)^l=0,
\end{equation}
which is equal to equations in \eqref{eq: Eq_m Th}. Thus, if $p$ equations in \eqref{eq: Eq m v5} hold $\forall m\in [1,p]$, we get $Q_m(\omega)=0$ and consequently $\Pi(\omega)=\alpha_p\beta_p\omega^{2p+1}$, where results in
\begin{equation}
    \label{eq: mag Upsilon v7}
\lim_{\omega \to 0}\Upsilon(\omega)=\frac{\alpha_p\beta_p}{\alpha_0\beta_0}\lim_{\omega \to 0}\omega^{2p+1}.
\end{equation}
Finally, substituting \eqref{eq: mag Upsilon v7} in \eqref{eq: var upsilon 0 pf1}, gives
\begin{equation}
    \label{eq: var upsilon 0 pf2}
    \upsilon_0=\lim_{\omega \to 0} 
\frac{d\,\log_{10}\left|\frac{\alpha_p\beta_p}{\alpha_0\beta_0}\omega^{2p+1}\right|}{d\,\log_{10}\omega},
\end{equation}
where results in
\begin{equation}
    \label{eq: upsilon proof v end}
    \upsilon_0=2p+1,
\end{equation}
and therefore
\begin{equation}
    h_0 = 20(\upsilon_0+1)=40(p+1)\ \mathrm{dB/dec}.
\end{equation}
\qed

%% file: Sections/07Appendix_A.tex
\section{Practical Shaping Filter Design Workflow}
\label{app:practical_workflow}

The proposed shaping-filter design can be implemented as follows:

\begin{itemize}
    \item Select the GFORE parameters, including \(\omega_r\), \(A_\rho\),
    and the required CgLp phase contribution around the target bandwidth.

    \item Choose the shaping-filter order \(p\) according to the desired
    low-frequency HOSIDF attenuation. From
    Theorem~\ref{theorem: shaping filter},
    \[
        h_0 = 40(p+1)\ \mathrm{dB/dec}.
    \]
    A larger \(p\) gives stronger harmonic suppression, but increases the
    filter order.

    \item Define the shaping filter as
    \[
    C_s(s)=
    \frac{\sum_{k=0}^{p} a_k \omega_r^{-k}s^k}
         {\sum_{k=0}^{p} b_k \omega_r^{-k}s^k},
    \]
    with \(a_0=b_0=1\).

    \item Select \(b_1,\ldots,b_p\) such that the denominator is Hurwitz and
    the filter provides the desired phase and attenuation behavior.

    \item Compute \(a_1,\ldots,a_p\) by solving the \(p\) algebraic equations
    in Theorem~\ref{theorem: shaping filter}.

    \item Check that no pole-zero cancellation reduces the effective filter
    order.

    \item Verify the design using \(H_1(\omega)\), the HOSIDFs
    \(H_n(\omega)\), \(n\geq3\), the pseudo-sensitivity
    \(S_\infty(\omega)\), and the number of reset instants \(N_r(\omega)\).
    If needed, adjust \(p\) or the denominator coefficients and repeat the
    procedure.
\end{itemize}